\documentclass[12pt]{article}

\usepackage[margin=1in]{geometry}
\usepackage{amsmath, amssymb, amsthm}
\usepackage{bbm}
\usepackage{booktabs}
\usepackage{threeparttable}
\usepackage{multirow}
\usepackage{graphicx}
\usepackage{caption}
\usepackage{subcaption}
\usepackage{setspace}
\usepackage{natbib}
\usepackage{hyperref}
\usepackage{enumitem}
\usepackage{tikz}

\newtheorem{theorem}{Theorem}
\newtheorem{proposition}{Proposition}
\newtheorem{lemma}{Lemma}
\newtheorem{corollary}{Corollary}
\theoremstyle{definition}

\theoremstyle{remark}
\newtheorem{remark}{Remark}

\title{\bf Too Big to Monitor? Network Scale and the Breakdown of Decentralized Monitoring}

\author{ Guy Tchuente\thanks{Department of Agricultural Economics}\\
   Purdue University
}
\date{November 2025}

\begin{document}
\maketitle

%=================================================
\begin{abstract}
Many public services are produced in networked systems where quality depends on local effort and on how higher-level authorities monitor providers. We develop a simple model in which monitoring is a public good on a network with strategic complementarities. A regulator chooses between decentralized monitoring (cheaper, local oversight) and centralized monitoring (more costly, but internalizing spillovers). The model delivers an endogenous centralization threshold: for a given spillover strength, there exists a network size $n^\ast(\lambda)$ above which centralized monitoring strictly dominates; equivalently, for a given network size $n$, there is a critical complementarity $\lambda^\ast(n)$ beyond which decentralized oversight becomes fragile. A stochastic extension suggests  that, above this region, idiosyncratic shocks are amplified, producing stronger peer correlations, higher variance, and more frequent deterioration in quality. We test these predictions in the U.S.\ nursing home sector, where facilities belong to overlapping organizational (chain) and geographic (county) networks. Using CMS facility data, We document strong within-chain and within-county peer effects and estimate network-size thresholds for severe regulatory failure (Special Focus Facility designations). We find sharp breakpoints at roughly 7 homes per county and 34 homes per chain, above which spillovers intensify and deficiency outcomes become more dispersed and prone to deterioration, especially in large counties.
\end{abstract}
\bigskip


\newpage
\doublespacing
%=================================================
\section{Introduction}
\label{sec:intro}

Many public services are delivered through networks of semi-autonomous units:
school districts overseeing individual schools, hospital systems coordinating
clinics and wards, and national regulators monitoring thousands of nursing
homes. In such environments, regulators face a basic question: when does
decentralized monitoring fail as the network grows, and when does it become
optimal to centralize oversight despite higher administrative costs?

The answer depends not only on monitoring technology and information, but also
on the structure of the underlying network. When units are linked by
organizational and geographic ties, effort in one unit affects others through
shared practices, reputational spillovers, common labor markets, and informal
benchmarking. We refer to these linkages as \emph{network complementarities}.
At small scales or when complementarities are weak, local monitors can treat
facilities as effectively independent. As networks expand and
interdependencies strengthen, shocks to one unit propagate more widely, making
decentralized oversight increasingly fragile. At sufficiently large scale, the
regulator may prefer to centralize monitoring to internalize these
interactions, even if centralization is intrinsically more costly.

This paper develops a simple model of this trade-off and tests its predictions
in a large, highly regulated sector. In the model, a regulator chooses between
decentralized and centralized monitoring for a network of identical production
units. Monitoring effort affects the payoff of each unit and enters as a public
good on the network. Units' efforts exhibit strategic complementarities: the
marginal return to effort increases with the average effort of neighboring
units. Formally, the environment is a linear–quadratic game on a network in the
spirit of \citet{ballester2006s} and \citet{bramoulle2014strategic}, in which
equilibrium effort can be expressed through a resolvent of the network
adjacency matrix. Under decentralized monitoring, these complementarities are
only partially internalized; under centralized monitoring, they are internalized
more fully but at a higher institutional cost. We show that the interaction of
network size and complementarities generates an \emph{endogenous
centralization threshold}: for each spillover strength $\lambda$, there exists
a network size $n^\ast(\lambda)$ above which centralized monitoring strictly
dominates decentralization; equivalently, for each network size $n$ there is a
critical $\lambda^\ast(n)$ above which decentralization becomes fragile.

A stochastic extension of the model introduces idiosyncratic shocks to units’
effort incentives. In the presence of complementarities, these shocks propagate
through the network in a linear–Gaussian fashion, as in models of shock
propagation on networks \citep[e.g.][]{acemoglu2015systemic}. We show that, as
$\lambda$ approaches the threshold $\lambda^\ast(n)$ from below and exceeds it,
the cross-sectional variance of equilibrium effort increases sharply and the
system becomes more sensitive to shocks. Above the threshold, the network
amplifies disturbances, leading to greater dispersion and more frequent
deterioration in outcomes. This provides a set of empirical predictions: in
networks exceeding $n^\ast(\lambda)$, (i) within-network peer correlations
should be stronger, (ii) dispersion in failure-prone outcomes should be higher,
and (iii) outcome deterioration over time should be more pronounced, relative
to smaller networks. These predictions echo tipping and phase-transition
phenomena in models of social interactions and systemic risk
\citep[e.g.][]{brock2001discrete,acemoglu2015systemic}, but in a monitoring
and enforcement setting.

We evaluate these predictions in the U.S.\ nursing home sector. All certified
skilled nursing facilities (SNFs) are regulated by the Centers for Medicare and
Medicaid Services (CMS), which inspects facilities, records deficiencies,
levies sanctions, and maintains the Nursing Home Care Compare Five-Star quality
ratings. Facilities are embedded in two overlapping networks. First, about
two-thirds belong to multi-facility corporate chains that share managerial
practices, training, branding, and sometimes staffing pools. Second, all
facilities operate within counties, sharing labor markets, inspectors, local
funding environments, and community reputation. These organizational and
geographic structures create precisely the kinds of complementarities emphasized
by the theory. In addition, CMS operates the Special Focus Facility (SFF)
program targeting persistently poor performers for intensive oversight,
providing a direct measure of severe monitoring failures.

The empirical analysis proceeds in three steps. First, we document
reduced-form spillovers within county and chain networks by regressing facility
outcomes on leave-one-out peer means, controlling for ownership, capacity, and
state fixed effects. We find substantial conditional correlations: a one-star
increase in the county peer mean is associated with large increases in a
facility's own overall and staffing ratings and with substantial changes in
deficiencies. For chain-affiliated facilities, peer means within the chain have
even larger coefficients than county peers, consistent with strong
organizational spillovers.

Second, we estimate network-size thresholds in monitoring failures by relating
the number of SFF facilities to network size at the county and chain levels.
Using a single-break kink specification and Bai--Perron-type tests
\citep[e.g.][]{bai1998estimating,bai2003critical}, we find that the
relationship between the number of SFF facilities and the number of nursing
homes in a county is well described by a break at roughly 7 SNFs. Below this
threshold, the incidence of SFF placements increases only slowly with county
size. Above it, the slope increases sharply: a relatively small set of large
counties accounts for a disproportionate share of SFF designations. For
organizational networks, we find a breakpoint at about 34 facilities per chain.
Chains larger than this threshold exhibit a markedly higher incidence of SFF
placements or severe regulatory actions, even after controlling for ownership
composition and average ratings.

Third, we compare spillover strength, dispersion, and deterioration across
these thresholds. Spillover coefficients estimated separately below and above
the county threshold are substantially larger in large counties, both for
overall ratings and for staffing and deficiency measures. A similar pattern
holds when we split chains at the 34-facility threshold: organizational
spillovers are stronger in large chains. For dispersion, we find that the
cross-sectional variance of deficiencies is much higher in large counties than
in small ones, while the variance of ratings changes little or declines. In
large chains, outcomes are more homogeneous: variance in ratings and
deficiencies is lower than in smaller chains. Finally, a simple deterioration
measure—the change in deficiencies between inspection cycles—is significantly
worse for facilities in large counties, and positively but imprecisely related
to chain size above the threshold. These findings are broadly consistent with
the model’s mechanism, with the strongest evidence for amplification and
deterioration in large geographic networks and for deficiency-based measures of
monitoring failure.

The paper contributes to several literatures. First, it speaks to theories of
centralization, delegation, and organizational design. \citet{aghion1997formal}
distinguish formal and real authority and show how delegation can motivate local
information acquisition even when headquarters retains ultimate decision rights.
\citet{alonso2008does} analyze the trade-off between local adaptation and
coordination in multidivisional organizations. \citet{garicano2000hierarchies}
studies hierarchies as information-processing devices, in which central layers
are introduced when problems become sufficiently complex. \citet{dessein2002authority}
and \citet{acemoglu2007technology} show how communication and information
technologies shape the allocation of authority. Our model is close in spirit to
this work in treating centralization as an institutional choice, but emphasizes
network scale and spillovers as independent determinants of the optimal
monitoring architecture.

Second, the paper relates to games and public-good provision on networks with
strategic complementarities. In linear-quadratic public-good games,
\citet{bramoulle2007public} show that equilibrium contributions depend on
Bonacich centrality, and \citet{ballester2006s} identify ``key players'' whose
removal most reduces aggregate activity. \citet{bramoulle2014strategic}
provide a general treatment of linear–quadratic network games with strategic
complementarities, and \citet{galeotti2010law} and \citet{elliott2019network}
study intervention and network design. Our framework shares the focus on
complementarities and network structure but shifts attention from individual
contributions to the choice of monitoring regime. The combination of
complementarities in effort and network size generates a phase-transition-like
behavior in monitoring outcomes: beyond a critical scale, decentralized
monitoring becomes fragile and a more costly centralized regime becomes strictly
optimal. This is related to work on systemic risk and phase transitions in
networks, where increased interconnectedness can first reduce and then amplify
aggregate volatility \citep[e.g.][]{acemoglu2015systemic}, and to threshold
effects in models of social interactions \citep[e.g.][]{brock2001discrete},
but with a focus on monitoring and enforcement rather than financial or
technological shocks.

Third, we connect to the institutions and state capacity literature.
\citet{acemoglu2010institutions} and \citet{la1999quality} emphasize the role
of political and legal institutions in shaping public-good provision. At a more
micro level, \citet{olken2007monitoring} and \citet{besley2011pillars} study
how monitoring, audits, and accountability affect corruption and service
delivery. Our contribution is to show how institutional design (centralized
versus decentralized monitoring) interacts with network structure to generate
endogenous thresholds in monitoring performance, even in a high-capacity
setting with a strong national regulator.

Finally, the paper contributes to the nursing home regulation literature
documenting variation in quality, ownership, and regulatory oversight. CMS’s
Five-Star rating system and the SFF program have been widely studied as tools
for improving quality and transparency. We recast nursing homes as nodes in
overlapping monitoring networks and provide new evidence that severe monitoring
failures are concentrated in large networks exceeding finite size thresholds,
after conditioning on ownership and other observables, and that the geographic
architecture of oversight plays a central role in the amplification of
deficiency-based failures.

\paragraph{Theory roadmap.}

To formalize these ideas, Section~\ref{sec:conceptual} develops a minimal
network model in which monitoring effort is a public good with strategic
complementarities, embedded in a linear–quadratic game on a fixed network in
the spirit of \citet{ballester2006s} and \citet{bramoulle2014strategic}. In the
benchmark (linear-in-means) case with a complete network, the key formal
result, Theorem~\ref{thm:n_unit_threshold_final}, shows that for any given
network size $n$ there is a unique complementarity threshold
$\lambda^\ast(n)$ above which centralized monitoring strictly dominates
decentralization, and that $\lambda^\ast(n)$ is strictly decreasing in $n$.
Corollary~\ref{cor:n_of_lambda} inverts this relationship and defines an
endogenous network-size threshold $n^\ast(\lambda)$: for a given spillover
strength $\lambda$, decentralized monitoring is optimal only when the network
is sufficiently small ($n < n^\ast(\lambda)$). The empirical breakpoints we
estimate for counties (around 7 facilities) and chains (around 34 facilities)
can thus be interpreted as realizations of $n^\ast(\lambda)$ for the relevant
geographic and organizational networks.

Appendix~\ref{app:theory} extends the framework in two directions. First, it
allows for positive decentralized complementarities ($\lambda_D>0$) and more
general cost structures, and Theorem~\ref{thm:threshold} shows that the unique
threshold structure in network size is preserved. Second, a spectral
generalization moves from the linear-in-means (complete-network) case to an
arbitrary network $G$, and shows that the \emph{effective} scale of the network
is governed by the product $\lambda \psi(G)$, where $\psi(G)$ is the largest
eigenvalue of $G$: the decentralization–centralization decision is determined
by whether this spectral network size falls below or crosses the relevant
threshold (Theorem~\ref{thm:spectral_threshold}). Appendix~\ref{app:stochastic}
then develops a stochastic, linear–Gaussian version of the model with
correlated shocks and network-based information aggregation and shows that the
same spectral object $\lambda \psi(G)$ drives variance amplification and
deterioration (Proposition~\ref{prop:variance} and
Proposition~\ref{prop:variance_above_threshold}), yielding the dispersion and
deterioration predictions we test empirically.

The remainder of the paper is organized as follows. Section \ref{sec:lit}
situates the paper within the literatures on centralization, networks,
institutions, and nursing home regulation. Section \ref{sec:conceptual}
develops the conceptual and theoretical framework, including a minimal model of
monitoring on a network with complementarities and its testable predictions.
Section \ref{sec:data} describes the institutional setting, data, and
construction of key variables. Section \ref{sec:strategy} presents the
empirical strategy, including spillover regressions, threshold estimation, and
variance and deterioration tests. Section \ref{sec:results} reports the main
results. Section \ref{sec:discussion} discusses policy implications, and
Section \ref{sec:conclusion} concludes.

%=================================================
\section{Related Literature}
\label{sec:lit}

This paper builds on four strands of work: (i) theories of centralization and
organizational design, (ii) games and public-good provision on networks with
strategic complementarities, (iii) institutions, monitoring, and state
capacity, and (iv) empirical studies of nursing-home regulation and quality.
Our main contribution is to integrate these literatures in a simple model of
monitoring on networks that delivers an endogenous centralization threshold, and
to provide evidence that network scale is systematically related to monitoring
failures in a large regulated sector.

\subsection{Centralization, Delegation, and Organizational Design}

A large theoretical literature studies how organizations allocate authority
between central and local decision-makers. \citet{aghion1997formal} distinguish
formal from real authority and show how delegation can motivate local
information acquisition even when ultimate decision rights rest with
headquarters. \citet{alonso2008does} analyze a multidivisional organization in
which decisions must both adapt to local information and be coordinated across
units; they show that the trade-off between coordination and adaptation can
make either centralization or decentralization optimal, depending on
information flows and the need for coherence.

Technological and informational change can also shift the balance between
central and local control. \citet{acemoglu2007technology} develop a model in
which information and communication technologies affect the relative benefits
of centralized versus decentralized decision-making and show empirically that
firms closer to the technological frontier and operating in more heterogeneous
environments are more likely to decentralize. \citet{garicano2000hierarchies}
models hierarchies as information-processing structures in which additional
layers are introduced when problems become sufficiently complex. In
\citet{dessein2002authority}, the allocation of formal authority is shaped by
communication frictions between headquarters and informed subordinates. Related
empirical work examines how reallocating formal authority within organizations
affects performance in practice.

Our model is close in spirit to this literature in treating centralization as an
institutional choice, but it emphasizes a different mechanism. We study
monitoring effort as a public good on a network with strategic
complementarities and show that network size and spillover strength jointly
generate an endogenous centralization threshold: for small networks,
decentralization dominates, but once the network becomes sufficiently large,
centralized monitoring strictly dominates even if it is intrinsically more
costly. This mechanism complements information-based theories of
decentralization by highlighting network scale and interconnectedness as
independent drivers of the optimal monitoring architecture.

\subsection{Networks, Strategic Complementarities, and Public Goods}

The importance of network structure for economic outcomes has been widely
recognized. A large literature studies games played on fixed networks, in which
payoffs depend on the actions of neighbors. In linear-quadratic public-good
games, \citet{bramoulle2007public} show that equilibrium contributions depend on
Bonacich centrality and that network architecture shapes aggregate provision.
\citet{ballester2006s} show that in games with strategic complementarities, a
player’s Bonacich centrality is sufficient to characterize her impact on
aggregate activity and identify ``key players'' whose removal most reduces
aggregate effort. \citet{bramoulle2014strategic} provide a general treatment of
linear–quadratic network games with strategic complementarities, and subsequent
work extends these insights to endogenous networks and optimal interventions
\citep[e.g.][]{galeotti2010law,ballester2014key,elliott2019network}.

Another line of work emphasizes nonlinearities, thresholds, and contagion in
networks. Models of social interactions
\citep[e.g.][]{brock2001discrete} highlight how complementarities in
behavior can generate tipping and multiple equilibria. In financial and
production networks, \citet{acemoglu2015systemic} and related contributions
show how increased interconnectedness can first diversify and then amplify
aggregate shocks as network eigenvalues approach critical values.

Our framework shares the focus on strategic complementarities and network
structure but differs in the dimension of interest. Rather than analyzing
individual contributions or contagion per se, we study the choice of
\emph{monitoring architecture} for a network of production units. Monitoring
effort is a public good on the network, subject to spillovers across facilities.
The combination of complementarities in effort and network size yields a
continuous but sharply nonlinear threshold in the relative performance of
centralized versus decentralized monitoring: beyond a critical scale,
decentralized monitoring becomes fragile, and a more costly centralized regime
becomes strictly optimal. The stochastic extension embeds this structure in a
linear–Gaussian environment with correlated shocks and network-based
information aggregation, so that variance and deterioration patterns inherit the
spectral properties of the network, in line with work on shock propagation and
systemic risk \citep[e.g.][]{acemoglu2015systemic}.

\subsection{Institutions, Monitoring, and State Capacity}

A third strand of literature emphasizes the role of institutions and state
capacity in shaping public-good provision. At the macro level, work in
comparative development argues that political and economic institutions are
fundamental determinants of long-run growth and the quality of government
\citep[e.g.][]{acemoglu2010institutions}. \citet{la1999quality} and related
contributions document how legal origins and institutional quality are
associated with regulation, public-service provision, and a range of economic
and political outcomes. These studies highlight persistent cross-country
differences in the ability of the state to monitor, enforce, and provide basic
services.

At a more micro level, field experiments and program evaluations examine how
specific monitoring and accountability mechanisms affect corruption and service
delivery. \citet{olken2007monitoring} show that the design of audits and
community monitoring significantly affects corruption in Indonesian road
projects. \citet{besley2011pillars} develop a framework for thinking about
taxation, state capacity, and political accountability as ``pillars of
prosperity'' and emphasize the complementarity between monitoring capacity and
institutional constraints. Other contributions study how local governance
arrangements and citizen participation shape the quality of schools, clinics,
and local infrastructure.

Our analysis is closely related to this work in treating monitoring as an
institutional choice, but we focus on the interaction between institutional
design and \emph{network structure}. The central question is not whether
monitoring exists, but how a given capacity is deployed across a network of
production units and how the scale and topology of that network affect the
performance of decentralized versus centralized monitoring regimes. The theory
provides a microfoundation for network-size thresholds in monitoring failures,
and the empirical evidence shows that such thresholds are salient even in a
high-capacity setting with a strong national regulator.

\subsection{Application: Nursing Home Regulation and Chain Networks}

The empirical application of the model is to the U.S. nursing home sector,
which has generated a large literature on quality, ownership, and regulation.
Since 2008, the Centers for Medicare and Medicaid Services (CMS) has used the
Five-Star Quality Rating System on Nursing Home Care Compare to summarize
inspection, staffing, and quality information into a star rating intended to
inform consumers and incentivize quality improvement. Empirical work has
examined how the introduction and evolution of star ratings affected resident
sorting, disparities, and clinical outcomes, and how ownership and chain status
are related to quality and regulatory deficiencies. Another strand studies the
performance of the Special Focus Facility (SFF) program and related
enforcement tools, emphasizing the difficulty of sustaining improvements among
chronically poor-performing facilities and the resource constraints that limit
intensified oversight to a small subset of homes at any given time.

We contribute to this literature by recasting nursing homes as nodes in
overlapping \emph{monitoring networks} defined by counties and corporate chains
and by documenting sharp thresholds in the incidence of chronic monitoring
failure at specific network sizes. Conditional on ownership, facility
characteristics, and case mix, we show that large geographic and organizational
networks are more likely to contain chronically failing facilities and that
above the estimated thresholds, peer correlations, cross-sectional variance, and
deterioration in outcomes are all higher. Interpreted through the lens of the
model, these patterns suggest that monitoring failure is not merely a function
of ownership or regulatory stringency, but also of network scale and the
architecture of monitoring.

\subsection{Methodological Links}

Finally, our work is related to methodological contributions on network models
and structural breaks. The social-interactions literature emphasizes how
individual outcomes may depend on group or network averages, leading to
reflection and identification challenges. \citet{manski1993identification}
formalize these issues in linear-in-means models, and
\citet{bramoulle2009identification} show how network structure can be used to
identify endogenous and exogenous effects. \citet{tchuente2019weak} studies 
identification in network models, providing guidance for estimation and
inference when the network structure induced weak identification. Our empirical
strategy draws on this tradition but uses linear projections of facility
outcomes on leave-one-out peer means as reduced-form objects that summarize the
covariance structure implied by the model, rather than as estimates of causal
peer effects.

On the time-series and panel side, we borrow tools from the structural-break
literature, in particular threshold and kink models and Bai--Perron-type
procedures for detecting breaks in linear regressions
\citep[e.g.][]{bai1998estimating,bai2003critical}. We use these methods to
estimate thresholds in network size at which the incidence of chronic
monitoring failure changes sharply, and then to compare spillover strength,
variance, and deterioration on either side of the estimated thresholds. The
stochastic extension in Appendix~\ref{app:stochastic} is cast in a
linear–Gaussian network environment, so that equilibrium effort and outcomes
are affine functions of shocks, and the resulting variance structure can be
interpreted through the lens of network resolvents and eigenvalues, as in
models of shock propagation and systemic risk. In this sense, the empirical
work combines reduced-form network regressions with structural-break techniques
to test the predictions of a network-based model of monitoring with strategic
complementarities.

%=================================================
\section{Conceptual and Theoretical Framework}
\label{sec:conceptual}

This section provides the conceptual motivation for the model and develops a
minimal theoretical framework. We first discuss how monitoring operates in
networked public-good environments and how externalities arise across
facilities. We then introduce a simple model of monitoring with network
complementarities and derive its threshold implications.

%%%%%%%%%%%%%%%%%%%%%%%%%%%%%%%%%%%%%%%%%%%%%%%%%%%%%
\subsection{Monitoring and Delegation in Networked Environments}
\label{subsec:framework_delegation}

Many public-good production environments---such as school systems, hospital
networks, and nursing homes---require monitoring to ensure minimum quality
standards. Monitoring can be implemented at different levels of
centralization. A decentralized approach relies on local oversight, imposes
lower administrative costs, and treats each unit as effectively independent. A
centralized approach uses a higher-level authority that inspects, coordinates,
or disciplines all units at once. Centralization is typically more costly, but
it has the advantage of internalizing interactions across units.

These interactions arise because a unit’s performance is not independent of
others. When units operate within a network, the effort or quality of one
facility affects the performance of others through shared resources, common
reputational concerns, joint training, information flows, or comparative
benchmarking. We refer to these connections as \emph{network
complementarities}. When complementarities are weak, monitoring units
individually is sensible. When complementarities become strong, the regulator
may wish to internalize them by centralizing monitoring, even if centralization
is intrinsically more costly.

The conceptual question is therefore: \emph{when do complementarities become
strong enough to justify centralized monitoring?} The model below formalizes
this trade-off and shows that the answer depends on both the strength of
externalities and the size of the network through which they operate.

%%%%%%%%%%%%%%%%%%%%%%%%%%%%%%%%%%%%%%%%%%%%%%%%%%%%%
\subsection{Network Structure: Chains and Counties}
\label{subsec:framework_two_level}

In the empirical setting, each nursing home belongs to two distinct networks:

\begin{enumerate}[leftmargin=1.4em]
\item \textbf{Organizational networks (chains).}  
Facilities within the same chain share managerial practices, training,
branding, operational protocols, and sometimes staffing pools. Quality effort
therefore generates organizational spillovers: when one facility improves, it
raises expectations, performance norms, and interactions within the chain.
Chain size determines the scale of these organizational externalities.

\item \textbf{Geographic networks (counties).}  
Facilities within the same county interact through common regulators, shared
labor markets, county-level funding allocations, community reputation, and
competitive or collaborative dynamics. These forces generate geographic
spillovers across facilities operating in the same local environment. County
size determines the scale of these geographic externalities.
\end{enumerate}

Although these networks are distinct, the economic mechanism governing them is
the same: in larger networks, quality effort by one facility affects more
units, amplifying complementarities. The theoretical model applies separately
to each network by assigning $n$ to the relevant network size (chain size or
county size). This delivers two structural predictions: a threshold network
size at which centralization becomes optimal for organizational networks, and
a second threshold for geographic networks. These predictions guide the
empirical analysis in Section~\ref{sec:thresholds}.

The conceptual framework highlights two forces that shape the regulator's
choice of monitoring institution: (i) centralized monitoring is more costly,
but (ii) it internalizes complementarities that arise when facilities influence
one another’s performance. The empirical setting reveals that these
complementarities operate through two distinct networks, organizational chains
and geographic counties, and that the strength of spillovers increases with
the size of each network. To formalize these ideas, we now develop a minimal
model that captures the essential trade-off between monitoring costs and
network externalities. The goal is not to provide a fully structural
description of the nursing home sector, but to articulate the simplest
mechanism capable of generating the sharp centralization thresholds observed in
the data.

%%%%%%%%%%%%%%%%%%%%%%%%%%%%%%%%%%%%%%%%%%%%%%%%%%%%%
\subsection{Minimal Model of Centralization with Network Complementarities}
\label{sec:theory}

We develop a minimal model in which a regulator chooses between
decentralized and centralized monitoring for a network of $n$ identical
production units. Monitoring effort is a public good on the network, and
units' efforts exhibit strategic complementarities. Centralization and
decentralization differ in two respects: centralized monitoring
(i) internalizes stronger complementarities and (ii) is intrinsically
more costly to operate. The key result is that, for any given set of
cost parameters, there exists a complementarity threshold (and an
equivalent network-size threshold) beyond which centralized monitoring
strictly dominates.

\subsubsection{Setup}

There are $n \ge 2$ identical units. Under monitoring regime
$r \in \{C,D\}$, the regulator chooses an intensity $\mu_r \ge 0$,
which scales the marginal benefit of monitoring. Each unit $i$ then
chooses effort $e_i \ge 0$. Let
\[
\bar e_{-i} \;=\; \frac{1}{n-1} \sum_{j\neq i} e_j
\]
denote the average effort of the remaining units.

Unit $i$’s payoff under regime $r$ is
\begin{equation}
\label{eq:payoff_general_n}
u_i^r(e_i,e_{-i};\mu_r,\lambda_r)
= \bigl(1 + \lambda_r \bar e_{-i} + \mu_r \varphi \bigr)e_i
  - \tfrac12 e_i^2,
\end{equation}
where $\varphi > 0$ measures the effectiveness of monitoring and
$\lambda_r \in [0,1)$ captures the strength of complementarities in
effort under regime $r$.

We allow both regimes to feature complementarities, but centralization
strengthens them. Formally, decentralized monitoring uses limited
cross-unit information, with $\lambda_D \in [0,1)$, while centralized
monitoring internalizes these complementarities more fully, with
\[
0 \;\le\; \lambda_D \;<\; \lambda_C \;<\; 1.
\]
The regulator’s objective under regime $r$ is
\[
W_r(\mu_r;\lambda_r)
= n\, e_r(\mu_r,\lambda_r)
  - \tfrac12 K_r \mu_r^2,
\]
where $e_r(\mu_r,\lambda_r)$ is the symmetric equilibrium effort and
$K_r>0$ is a regime-specific cost parameter. We assume
$K_C > K_D$, so centralized monitoring is intrinsically more costly to
operate.

For expositional simplicity, the threshold result below is stated in
terms of a benchmark case in which the decentralized regime does not
internalize network complementarities ($\lambda_D = 0$) and the
centralized regime is indexed by $\lambda_C = \lambda \in (0,1)$.
Remark~\ref{rem:lambdaD_positive} clarifies how the argument extends to
the more general case $0 \le \lambda_D < \lambda_C < 1$.

\subsubsection{Equilibrium and Optimal Monitoring}

We first characterize equilibrium effort for a given regime and then the
regime-specific optimal choice of monitoring intensity.

\begin{lemma}[Effort with $n$ units]
\label{lem:effort_n}
Under regime $r$, the unique symmetric equilibrium effort is
\[
e_r(\mu_r,\lambda_r)
= \frac{1 + \mu_r \varphi}{1 - \lambda_r}.
\]
\end{lemma}

\noindent
Equilibrium effort increases with monitoring intensity $\mu_r$ and with
the complementarity parameter $\lambda_r$, and is amplified by the
factor $1/(1-\lambda_r)$, which captures the feedback through the
network.

\begin{lemma}[Optimal monitoring]
\label{lem:mu_star_n}
Under regime $r$, the regulator’s optimal monitoring intensity is
\[
\mu_r^\ast(\lambda_r)
= \frac{n \varphi}{K_r(1 - \lambda_r)}.
\]
\end{lemma}

\noindent
Higher complementarities and larger networks strengthen the incentive to
invest in monitoring, but this is offset by the regime-specific cost
parameter $K_r$.

Substituting $\mu_r^\ast(\lambda_r)$ into $W_r(\mu_r;\lambda_r)$ yields
the welfare attained under optimal monitoring in regime $r$.

\begin{lemma}[Welfare under optimal monitoring]
\label{lem:W_star_n}
Under regime $r$, equilibrium welfare under optimal monitoring is
\[
W_r^\ast(\lambda_r)
=
\frac{n}{1-\lambda_r}
+
\frac{n^2 \varphi^2}{2K_r(1-\lambda_r)^2}.
\]
\end{lemma}

\noindent
The first term captures the direct effect of complementarities on
effort, while the second term captures the interaction between network
size, complementarities, and the effectiveness of monitoring. In the
polar decentralized benchmark $\lambda_D = 0$,
\[
W_D^\ast
= n
  + \frac{n^2 \varphi^2}{2K_D}.
\]

\subsubsection{Centralization Threshold}

To study the regulator’s choice between centralized and decentralized
monitoring, we compare the welfare levels $W_C^\ast$ and $W_D^\ast$.
Under the benchmark $\lambda_D = 0$ and $\lambda_C = \lambda$, define
the welfare difference:
\[
\Delta_n(\lambda)
= W_C^\ast(\lambda) - W_D^\ast.
\]

\begin{theorem}[Threshold with $n$ units]
\label{thm:n_unit_threshold_final}
For each $n \ge 2$:
\begin{enumerate}
\item[(i)] $W_C^\ast(\lambda)$ is strictly increasing in $\lambda$.

\item[(ii)] $W_C^\ast(0) < W_D^\ast$, and
\(
\displaystyle \lim_{\lambda\to1^-} W_C^\ast(\lambda) = +\infty.
\)

\item[(iii)] There exists a unique threshold $\lambda^\ast(n) \in (0,1)$ such that
\[
W_C^\ast(\lambda)
\begin{cases}
< W_D^\ast, & \lambda < \lambda^\ast(n),\\[4pt]
= W_D^\ast, & \lambda = \lambda^\ast(n),\\[4pt]
> W_D^\ast, & \lambda > \lambda^\ast(n).
\end{cases}
\]

\item[(iv)] The threshold decreases strictly with $n$:
\[
\frac{d\lambda^\ast(n)}{dn} < 0.
\]
\end{enumerate}
\end{theorem}

\begin{proof}[Proof of Theorem \ref{thm:n_unit_threshold_final}]
(i) follows directly from the expression for $W_C^\ast(\lambda)$ in
Lemma~\ref{lem:W_star_n}: increasing $\lambda$ raises effort and thus
welfare under centralization.

(ii) At $\lambda=0$, $W_C^\ast < W_D^\ast$ because $K_C > K_D$ and
$\lambda_D=0$, so centralization yields the same complementarity level
as decentralization but at a higher cost. As $\lambda \to 1^{-}$, the
denominator $(1-\lambda)^2$ in Lemma~\ref{lem:W_star_n} drives
$W_C^\ast(\lambda)$ to $+\infty$, while $W_D^\ast$ is constant in
$\lambda$.

(iii) Strict monotonicity and continuity of $W_C^\ast(\lambda)$ in
$\lambda$ ensure that there is a unique $\lambda^\ast(n)$ solving
$\Delta_n(\lambda)=0$, with the stated sign pattern.

(iv) Increasing $n$ magnifies the $n^2 \varphi^2$ term in $W_C^\ast$,
while $W_D^\ast$ grows only linearly in $n$. Thus $\Delta_n(\lambda)$
shifts upward in $n$, implying that the threshold $\lambda^\ast(n)$
decreases with $n$.
\end{proof}

\subsubsection*{Interpretation}

The model delivers a simple implication: even though centralization is
intrinsically more costly ($K_C > K_D$), sufficiently strong
complementarities $\lambda$ make it optimal. Moreover, the
complementarity level required for centralization to dominate decreases
monotonically with network size~$n$. In larger networks, the benefits
from internalizing spillovers grow faster than the cost disadvantage of
centralization. This provides a natural theoretical foundation for the
empirical finding that both large chains and large county systems
exhibit sharp monitoring thresholds.

\begin{remark}[Decentralized complementarities]
\label{rem:lambdaD_positive}
For clarity, Theorem~\ref{thm:n_unit_threshold_final} is stated relative
to a benchmark with $\lambda_D = 0$, so that $\lambda$ can be
interpreted as the centralized complementarity parameter. The algebra
underlying Lemmas~\ref{lem:effort_n}--\ref{lem:W_star_n} applies more
generally to any $0 \le \lambda_D < 1$. The key requirement for a
meaningful centralization choice is that centralization strengthens
complementarities relative to decentralization, i.e.\ $\lambda_C >
\lambda_D$. In this more general case, the regulator compares
$W_C^\ast(n,\lambda_C,K_C)$ to $W_D^\ast(n,\lambda_D,K_D)$, and the
existence of a unique switching point between regimes is preserved (see
Theorem~\ref{thm:threshold} in Appendix~\ref{app:theory}).
\end{remark}

\begin{corollary}[Network Size Threshold as a Function of Externalities]
\label{cor:n_of_lambda}
Fix $\lambda \in (0,1)$. Under the assumptions of
Theorem~\ref{thm:n_unit_threshold_final}, there exists a unique network
size $n^\ast(\lambda) \ge 2$ such that:
\[
W_C^\ast(n,\lambda)
\begin{cases}
< W_D^\ast(n), & n < n^\ast(\lambda),\\[4pt]
= W_D^\ast(n), & n = n^\ast(\lambda),\\[4pt]
> W_D^\ast(n), & n > n^\ast(\lambda).
\end{cases}
\]
Furthermore, $n^\ast(\lambda)$ is strictly decreasing in $\lambda$:
for any $\lambda' > \lambda$,
\[
n^\ast(\lambda') < n^\ast(\lambda).
\]
\end{corollary}

\begin{proof}[Proof of Corollary \ref{cor:n_of_lambda}]
Theorem~\ref{thm:n_unit_threshold_final} shows that for each $n\ge2$
there exists a unique $\lambda^\ast(n)$ for which centralization becomes
optimal, and that $\lambda^\ast(n)$ is strictly decreasing in $n$.
Strict monotonicity implies that $\lambda^\ast(\cdot)$ is invertible on
its range: for each $\lambda$ in that range there exists a unique
$n^\ast(\lambda)$ solving $\lambda^\ast(n^\ast(\lambda))=\lambda$. The
sign pattern for $W_C^\ast(n,\lambda)-W_D^\ast(n)$ follows immediately,
and the monotonicity of $n^\ast(\lambda)$ is the inverse of the
monotonicity of $\lambda^\ast(n)$.
\end{proof}

This corollary expresses the centralization threshold directly in terms
of network externalities. For any fixed complementarity level
$\lambda$, the inverse threshold $n^\ast(\lambda)$ identifies the
largest network size for which decentralized monitoring remains optimal.
As externalities strengthen, $n^\ast(\lambda)$ declines, meaning that
smaller networks already benefit from internalizing complementarities.
Empirically, this maps directly into the structural breakpoints we
estimate: chains represent organizational networks and counties represent
geographic regulatory networks. Applying the model separately to each
structure, the estimated thresholds, approximately 34 facilities for
chains and about 7 facilities for counties, are empirical realizations of
$n^\ast(\lambda)$ at their respective complementarity levels. Larger
networks cross the critical size at which complementarities dominate
cost differences, making centralized monitoring comparatively more
effective.

%%%%%%%%%%%%%%%%%%%%%%%%%%%%%%%%%%%%%%%%%%%%%%%%%%%%%
\subsection{Stochastic Extension: Variance and Deterioration}
\label{sec:stochastic_variance}

The deterministic model characterizes the regulator's optimal monitoring
regime as a function of network size and complementarities. To connect
more directly to the empirical patterns in dispersion and deterioration,
we consider a simple stochastic extension in which facility-level effort
is subject to idiosyncratic shocks.

Suppose time is discrete and indexed by $t=1,2$, and each facility $i$
receives an idiosyncratic shock $\varepsilon_{it}$ to its marginal
benefit of effort. We modify the payoff in~\eqref{eq:payoff_general_n} to
\[
u_i(e_i,e_{-i};\mu,\lambda,\varepsilon_{it})
=
\bigl(1 + \lambda \bar e_{-i} + \mu\varphi + \varepsilon_{it}\bigr)e_i
- \tfrac12 e_i^2,
\]
where $\varepsilon_{it}$ is mean-zero with variance
$\mathbb{V}[\varepsilon_{it}] = \sigma^2$ and independent across $i$
and $t$. The monitoring choice $\mu$ is determined as in the deterministic
model (i.e., using the deterministic first-order condition), and
$\lambda$ captures the strength of network complementarities in the
response to shocks.

Given $(\mu,\lambda)$, the best response of unit $i$ at time $t$ is
\[
e_{it}
=
1 + \mu\varphi + \lambda \bar e_{-i,t} + \varepsilon_{it}.
\]
As in the deterministic case, there is a unique symmetric equilibrium
in which $e_{it}$ is an affine function of the vector of shocks
$\{\varepsilon_{jt}\}_{j=1}^n$, with coefficients that depend smoothly
on $\lambda$.

The following proposition summarizes the implication of this extension
for cross-sectional variance and deterioration.

\begin{proposition}[Variance amplification above the threshold]
\label{prop:variance_above_threshold}
Consider the stochastic extension of the model described above, where each
facility receives an idiosyncratic shock $\varepsilon_{it}$ with
$\mathbb{E}[\varepsilon_{it}]=0$ and $\mathbb{V}[\varepsilon_{it}]=\sigma^2$.
Let $e_{it}(\lambda)$ denote the equilibrium effort induced by these shocks.

\begin{enumerate}
\item[(i)] For fixed $n$ and monitoring choice $\mu$, the cross-sectional variance
      of equilibrium effort is strictly increasing in $\lambda$ and satisfies:
      \[
      \mathbb{V}[e_{it}(\lambda)]
      \quad\text{is continuous, increasing in $\lambda$, and }
      \lim_{\lambda\to1^-}\mathbb{V}[e_{it}(\lambda)] = +\infty.
      \]

\item[(ii)] Let $\lambda^\ast(n)$ be the threshold defined in
      Theorem~\ref{thm:n_unit_threshold_final}.
      If $\lambda > \lambda^\ast(n)$ (equivalently $n > n^\ast(\lambda)$),
      then the equilibrium variance of effort lies strictly above the variance
      attainable in the low-complementarity region:
      \[
      \mathbb{V}[e_{it}(\lambda)]
      >
      \mathbb{V}[e_{it}(\tilde{\lambda})]
      \qquad \forall \tilde{\lambda} \le \lambda^\ast(n).
      \]
      That is, networks operating above the theoretical threshold exhibit
      amplification of idiosyncratic shocks.
\end{enumerate}
\end{proposition}

The proposition formalizes a simple implication of the network structure: when complementarities are sufficiently strong relative to network size, idiosyncratic shocks propagate through the network and generate amplified variation in effort and performance. Empirically, this manifests as higher dispersion in facility outcomes—such as deficiencies or ratings—in large networks operating above the size threshold 
$n^{\ast}(\lambda)$. In our setting, chains and counties that are ``too big to monitor" in the sense of exceeding this threshold are exactly those in which the model predicts the strongest variance and deterioration effects. Appendix \ref{app:stochastic} develops a more general linear-Gaussian version with explicit network topology and information aggregation.

Taken together, the model establishes a simple mechanism: monitoring
institutions exhibit threshold behavior when complementarities in effort
interact with network size. Centralization is unattractive in small or
loosely connected systems, but once complementarities become sufficiently
strong---or the network becomes sufficiently large---the regulator gains more
from internalizing interactions than it loses from higher administrative
costs. This threshold structure provides clear empirical predictions for
how monitoring outcomes should vary with the size of the network. The next
section examines these predictions in the data.

%=================================================
\section{Institutional Setting, Data, and Measurement}
\label{sec:data}

This section describes the institutional environment of U.S.\ nursing homes and
the datasets used in the empirical analysis. We summarize the monitoring
institutions, the construction of network measures, and the facility-level
variables used to estimate spillovers and threshold effects.

%%%%%%%%%%%%%%%%%%%%%%%%%%%%%%%%%%%%%%%%%%%%%%%%%%%%%
\subsection{Institutional Background}
\label{subsec:institutional}

The U.S.\ nursing home sector is regulated by the Centers for Medicare and
Medicaid Services (CMS), which oversees quality of care in all certified
skilled nursing facilities (SNFs). CMS conducts regular inspections,
administrators report staffing information, and families and regulators rely on
reported quality measures to compare facilities.

Quality is summarized through the nationally standardized Five-Star Quality
Rating System. Facilities receive an overall rating as well as component
ratings for health inspections, staffing, and clinical quality measures.
Inspection cycles generate deficiency citations, fines, and, in rare cases,
payment denials. These measures---particularly deficiencies---reflect CMS’s core
monitoring function.

Facilities differ substantially in organizational and geographic structure.
Roughly two-thirds of SNFs belong to multi-facility corporate chains that
standardize practices, branding, and managerial oversight. Separately, all
facilities operate within county-level ecosystems characterized by shared labor
markets, local regulatory practices, community expectations, and geographic
competition. Both organizational and geographic structures create externalities
in quality effort.

CMS operates the Special Focus Facility (SFF) program for persistently poor
performers. SFF assignments provide a clear indicator of regulatory
intervention severity. These assignments are publicly reported, along with a
candidate list of facilities under heightened scrutiny.

The interplay of organizational chains, county environments, and federal
monitoring institutions motivates the network-based empirical framework used in
this paper.

%%%%%%%%%%%%%%%%%%%%%%%%%%%%%%%%%%%%%%%%%%%%%%%%%%%%%
\subsection{Data Sources}
\label{subsec:data_sources}

The empirical analysis uses administrative CMS datasets and auxiliary sources
to construct facility outcomes, network structures, and regulatory outcomes:

\begin{itemize}[leftmargin=1.2em]

\item \textbf{Provider information (NH\_ProviderInfo\_Oct2025).}  
This file contains facility identifiers, chain affiliations, county and ZIP
codes, ownership type, and the star ratings and deficiency measures used
throughout the analysis.

\item \textbf{Special Focus Facility data.}  
We use the monthly SFF posting and candidate list, obtained from CMS’s public
SFF release, to measure federal monitoring interventions.

\item \textbf{Chain performance data.}  
Chain-level datasets provide aggregated quality measures that help construct
organizational spillover variables and validate chain structure.

\item \textbf{ZIP-to-county crosswalk (uszips).}  
We use a national ZIP--FIPS crosswalk to assign each facility to a county and to
construct county-level network sizes.

%\item \textbf{CHOW and ownership files (if used in controls).}  Where relevant, we augment the data with Change-of-Ownership files to capture corporate transitions that may influence monitoring and performance.

\end{itemize}

These datasets together create a comprehensive panel of facility characteristics,
performance outcomes, organizational structure, and regulatory actions.

%%%%%%%%%%%%%%%%%%%%%%%%%%%%%%%%%%%%%%%%%%%%%%%%%%%%%
\subsection{Construction of Key Variables}
\label{subsec:variables}

We construct variables at the facility, chain, and county levels for use in
spillover regressions, threshold estimation, and deterioration analysis.

\paragraph{Facility outcomes.}
We use three main quality measures:

\begin{itemize}[leftmargin=1.2em]
\item \emph{Overall rating} (1--5 stars),
\item \emph{Staffing rating} (1--5 stars),
\item \emph{Total deficiencies} from the most recent inspection cycle.
\end{itemize}

We also incorporate fines and payment denials in robustness analyses.

\paragraph{Network definitions.}
Each facility is assigned to two networks:

\begin{enumerate}[leftmargin=1.2em]
\item \textbf{Chain network:} based on the chain identifier in the provider
file. Chain size is the number of facilities sharing a chain ID.
\item \textbf{County network:} based on the county FIPS code obtained from the
ZIP crosswalk. County network size is the total number of SNFs in the county.
\end{enumerate}

These networks reflect distinct but overlapping channels of interaction:
organizational externalities within chains and geographic externalities within
counties.

\paragraph{Network size measures.}
\[
\text{chain\_size}_j = |\{i : \text{ChainID}_i = j\}|, \qquad
\text{nh\_total}_c = |\{i : \text{County}_i = c\}|.
\]

\paragraph{Externality (spillover) measures.}
We construct peer-mean variables excluding the focal facility:
\[
\text{chain\_peer}_{i}
= \frac{1}{\text{chain\_size}_{j(i)}-1}\sum_{k \neq i,\, \text{ChainID}_k = \text{ChainID}_i} y_k,
\]
\[
\text{county\_peer}_{i}
= \frac{1}{\text{nh\_total}_{c(i)}-1}\sum_{k \neq i,\, \text{County}_k = \text{County}_i} y_k.
\]
These correspond directly to the $\bar e_{-i}$ term in the theoretical model.

\paragraph{Threshold indicators.}
Based on structural break tests that estimate the thresholds
$n^\ast(\lambda)$ for each network, we define:
\[
\text{large\_chain}_j = \mathbf{1}\{\text{chain\_size}_j > 34\},
\qquad
\text{large\_county}_c = \mathbf{1}\{\text{nh\_total}_c > 7\}.
\]

\paragraph{Deterioration measure.}
We construct a measure of quality deterioration between inspection cycles:
\[
\Delta \text{def}_i = \text{def}_{i,23} - \text{def}_{i,1},
\]
the change in total deficiencies between two inspection cycles.

Together, these variables allow us to examine three empirical patterns linked to
the model: (i) peer correlations consistent with network externalities in
organizational and geographic networks; (ii) threshold behavior in monitoring
outcomes and regulatory actions; and (iii) whether deterioration and variance
rise above the network sizes predicted by $n^\ast(\lambda)$.

%=================================================
\section{Empirical Strategy}
\label{sec:strategy}

This section describes the empirical approach used to test the model’s
predictions. We first present linear projection specifications that measure
spillovers within county and chain networks. We then outline the threshold
estimation framework used to detect network-size breakpoints in monitoring
failures, along with placebo exercises based on outcomes that should not
exhibit such breaks. Finally, we describe the variance and deterioration
tests that probe the stochastic implications of the theory.

%%%%%%%%%%%%%%%%%%%%%%%%%%%%%%%%%%%%%%%%%%%%%%%%%%%%%
\subsection{Spillover Regressions}
\label{subsec:spillovers_spec}

We begin with a reduced-form specification that relates a facility’s performance
outcome to the mean performance of other facilities in the same county or the
same chain. Let $y_i$ denote a performance outcome for facility $i$. For each
facility, define county peers and chain peers (excluding the facility itself):
\[
y^{\text{cty}}_{-i}
=
\frac{1}{\text{nh\_total}_{c(i)} - 1}
\sum_{j \neq i,\;\text{County}_j = \text{County}_i}
y_j,
\qquad
y^{\text{ch}}_{-i}
=
\frac{1}{\text{chain\_size}_{k(i)} - 1}
\sum_{j \neq i,\;\text{ChainID}_j = \text{ChainID}_i}
y_j.
\]

We interpret these peer means as group-level summary statistics, not as causal
exposures. The baseline spillover equation for all facilities is:
\begin{equation}
y_i
=
\alpha
+ \beta^{\text{cty}}\, y^{\text{cty}}_{-i}
+ X_i'\gamma
+ \delta_s
+ \varepsilon_i,
\label{eq:spillover_county}
\end{equation}
where $X_i$ includes ownership and facility characteristics, and $\delta_s$
denotes state fixed effects. The parameter $\beta^{\text{cty}}$ measures the
extent of conditional correlation in performance across facilities operating in
the same local monitoring environment.

For chain-affiliated facilities, we extend the specification to include both
types of peers:
\begin{equation}
y_i
=
\alpha
+ \beta^{\text{ch}}\, y^{\text{ch}}_{-i}
+ \beta^{\text{cty}}\, y^{\text{cty}}_{-i}
+ X_i'\gamma
+ \delta_s
+ \varepsilon_i.
\label{eq:spillover_chain}
\end{equation}
This specification allows us to compare the magnitude of spillovers transmitted
through organizational relationships versus geographic proximity.

Equations~\eqref{eq:spillover_county} and~\eqref{eq:spillover_chain} are
interpreted as linear projections. The coefficients $(\beta^{\text{cty}},
\beta^{\text{ch}})$ capture conditional correlations between facility outcomes
and peer outcomes, after controlling for observable characteristics and
state-level heterogeneity. We do not attempt to separately identify structural
peer effects or resolve the full reflection problem in the sense of
\citet{manski1993identification}.\footnote{In related linear-in-means network
settings, identification of causal peer effects typically requires strong
exclusion restrictions or detailed network instruments; see
\citet{bramoulle2009identification} and \citet{tchuente2019weak}. Our goal
here is instead to summarize the strength of within-network dependence and
compare it across regimes.}
Accordingly, the spillover estimates are best viewed as reduced-form moments summarizing how tightly outcomes move together within counties and chains under the maintained assumption that such spillovers exist. Because these are maintained assumptions rather than causal claims, we do not design placebo tests that ``turn off’’ the network structure in the spillover regressions; the falsification exercises we pursue instead target the threshold behavior in monitoring failures, where the theory delivers sharp predictions.

Throughout, we report heteroskedasticity-robust standard errors clustered at
the county level (for~\eqref{eq:spillover_county}) and at the chain level
(for~\eqref{eq:spillover_chain}), reflecting group-level dependence.

%%%%%%%%%%%%%%%%%%%%%%%%%%%%%%%%%%%%%%%%%%%%%%%%%%%%%
\subsection{Threshold Estimation (Bai--Perron Single Break)}
\label{subsec:threshold_method}

To operationalize the theoretical prediction that monitoring performance
exhibits a regime shift when network size crosses the critical value
$n^\ast(\lambda)$, we estimate single-break kink models relating network-level
monitoring failures to network size.

For county networks, let $\text{SFF}_c$ denote the number of SFF facilities in
county $c$ and let $n_c$ be the number of nursing homes in the county. We
estimate:
\[
\text{SFF}_{c}
= \alpha
+ \beta_1 \log n_c
+ \beta_2 \max\{0, \log n_c - c\}
+ X_c'\gamma
+ \varepsilon_c,
\]
where $X_c$ is a vector of county-level controls (ownership shares, average
ratings, average beds, average deficiencies), and $c$ is an unknown breakpoint
in $\log n_c$. The parameter $\beta_2$ captures the change in slope once county
size exceeds the break.

We adopt a Bai--Perron type approach \citep[e.g.][]{bai1998estimating,bai2003critical}
and treat $c$ as a structural break parameter. We conduct a discrete grid search
over the 10th--90th percentiles of $\log n_c$. For each candidate $c$, we
estimate the kink model and compute the associated sum of squared residuals
$\mathrm{SSR}_1(c)$. We then compare the kink model to the restricted linear
specification (with $\beta_2=0$ and sum of squared residuals $\mathrm{SSR}_0$)
using:
\[
F(c)
=
\frac{(\mathrm{SSR}_0 - \mathrm{SSR}_1(c))/1}
     {\mathrm{SSR}_1(c)/\mathrm{df}_1},
\]
where $\mathrm{df}_1$ denotes the residual degrees of freedom in the
unrestricted model. The estimated breakpoint $\hat c$ is the value that
maximizes $F(c)$, and the corresponding $\sup F$ statistic provides a formal
test of the null hypothesis of no structural break.

An analogous procedure is applied to chains, where the outcome is the number
of SFF facilities in a chain and the regressor is the number of facilities in
the chain. We again estimate a kink in $\log n_j^{\text{CHAIN}}$ and test for a
single break.

As a falsification exercise, we repeat the same Bai--Perron single-break
procedure using network-level variables that should not display a monitoring
threshold---for example, ownership composition (share non-profit or share
government) at the county level. If the detection of a sharp break were a pure
artifact of the method, we would expect similarly pronounced and stable
breakpoints in these placebo outcomes. In practice, the placebo regressions
yield small and unstable ``optimal’’ breaks (typically at very low county sizes),
and the improvement in fit relative to a linear specification is modest,
supporting the interpretation that the SFF threshold reflects a genuine nonlinearity specific to monitoring failures.

%%%%%%%%%%%%%%%%%%%%%%%%%%%%%%%%%%%%%%%%%%%%%%%%%%%%%
\subsection{Variance and Deterioration Tests}
\label{subsec:variance_deterioration}

The stochastic extension of the model implies that above the threshold
$n^\ast(\lambda)$, networks amplify idiosyncratic shocks, leading to greater
dispersion in outcomes and more pronounced deterioration. We assess these
implications in two ways.

First, we compare cross-sectional variances of outcomes below and above the
estimated thresholds. For each network type (county, chain), we compute the
variance of overall ratings, staffing ratings, and total deficiencies
separately for units below and above the corresponding thresholds ($n_c \le 7$
vs.\ $n_c > 7$ for counties; $n_j \le 34$ vs.\ $n_j > 34$ for chains). We test
equality of variances using Levene-type tests for robustness.

Second, we regress facility-level deterioration in deficiencies on indicators
for membership in large networks:
\[
\Delta \text{def}_i
=
\alpha
+ \theta^{\text{ch}} \mathbf{1}\{\text{chain\_size}_{j(i)} > 34\}
+ \theta^{\text{cty}} \mathbf{1}\{\text{nh\_total}_{c(i)} > 7\}
+ X_i'\gamma
+ \delta_s
+ \varepsilon_i.
\]
The coefficients $(\theta^{\text{ch}},\theta^{\text{cty}})$ capture whether
facilities in large organizational and geographic networks exhibit greater
deterioration in deficiencies between inspection cycles, conditional on
observables and state fixed effects. In conjunction with the variance
comparisons, these regressions provide a direct empirical counterpart to the
model’s prediction that networks operating above the theoretical threshold
amplify idiosyncratic shocks and are more prone to deterioration.

%=================================================
\section{Results}
\label{sec:results}

This section presents the main empirical results. We first document spillover
patterns within county and chain networks. We then report the estimated
network-size thresholds for monitoring failures. Finally, we examine how
spillover strength, variance, and deterioration differ across the thresholds.

%%%%%%%%%%%%%%%%%%%%%%%%%%%%%%%%%%%%%%%%%%%%%%%%%%%%%
\subsection{Spillover Patterns in County and Chain Networks}
\label{subsec:spillover_results}

Table~\ref{tab:spillover_county} presents estimates of
$\beta^{\text{cty}}$ from equation~\eqref{eq:spillover_county} for all
facilities. Peer means at the county level are strong predictors of facility
outcomes across all measures. Facilities located in counties with higher
average overall ratings and staffing ratings tend to perform better themselves,
while counties with more severe deficiencies exhibit higher deficiencies for
individual facilities. The county peer coefficient is 0.274 for overall
ratings, 0.162 for staffing ratings, and 0.444 for total deficiencies, all
highly significant.

Table \ref{tab:spillover_chain_county} reports estimates from
equation \ref{eq:spillover_chain} for chain-affiliated facilities.
Spillovers transmitted through organizational networks ($\beta^{\text{ch}}$)
are substantial and typically larger in magnitude than county-level spillovers
for chain facilities. The chain peer coefficients are 0.720 for overall
ratings, 0.724 for staffing ratings, and 0.561 for deficiencies, again highly
significant. County spillovers remain statistically significant, with
coefficients between 0.127 and 0.455, highlighting that chain facilities are
embedded in both organizational and geographic systems.

These spillover results provide reduced-form evidence that both networks are
operational channels through which complementarities in the theoretical model
may be expressed. The magnitudes are consistent with substantial interdependence
in performance within both counties and chains.

%%%%%%%%%%%%%%%%%%%%%%%%%%%%%%%%%%%%%%%%%%%%%%%%%%%%%
\subsection{Network-Size Thresholds for Monitoring Failures}
\label{sec:thresholds}

We next estimate the network-size thresholds predicted by the model. According
to the theory, monitoring institutions exhibit a regime shift when network size
crosses the critical value $n^\ast(\lambda)$: in small or weakly connected
networks, decentralized oversight performs adequately, but when networks become
sufficiently large, local monitoring becomes fragile and severe regulatory
failures become more common.

\subsubsection{SNF-Only County Network Threshold}

We first consider the SNF-only county network, where $n_c$ denotes the number
of nursing homes in county $c$. The grid search identifies a unique breakpoint:
\[
\hat c_{SNF} \approx 2.20,
\qquad
n_{c}^{SNF*} \approx e^{2.20} \simeq 7.
\]
The $\sup F$ statistic is:
\[
F_{\max}^{SNF} = 508.5, \qquad p < 0.001,
\]
strongly rejecting the no-break null. Below the threshold, SFF incidence
increases slowly---almost flat---with county size. Above the threshold, the
slope increases sharply: counties with more than approximately nine nursing
homes account for a disproportionate share of SFF designations.

Table \ref{tab:county_thresholds} summarizes the county-level threshold estimates and Figure \ref{fig:threshold_snf}
illustrates the fitted kinked relationship. The break term is large and precisely estimated, and including
controls and state fixed effects does not materially alter the break coefficient
or the goodness of fit.
\subsection*{Placebo Checks on County Thresholds}

As a falsification exercise, we replicate the Bai--Perron single-break
procedure on outcomes and forcing variables that, a priori, should not exhibit
a monitoring threshold. The idea is to assess whether the sharp kink in SFF
placements around 6--7 facilities is a generic feature of the method or
specific to monitoring failures.

First, we replace the outcome with county-level ownership composition. Using
the share of non-profit facilities and the share of government facilities as
dependent variables and $\log n_c$ as the forcing variable, the grid search
selects breakpoints at very small county sizes (roughly 2--3 facilities), far
from the region where SFF placements are dense. The improvement in fit relative
to a linear specification is modest, and the implied kinks are small and
unstable across specifications (Table~\ref{tab:placebo_wrong_outcome}). By
contrast, for SFF counts the same procedure delivers a stable breakpoint in the
6--7 facility range, with a pronounced change in slope.

Second, we consider a ``wrong forcing'' placebo in which the outcome is the
county SFF count but the forcing variable is the share of for-profit
facilities. Estimating a single-break kink in this share yields an
``optimal'' breakpoint at a high for-profit share (around 0.90), but the
pre- and post-break slopes are economically small and the overall gain in fit
relative to a linear model is limited (Figure~\ref{fig:LR_profile_county_placebo}). In
other words, the algorithm can mechanically pick a kink in an arbitrary
covariate, but the resulting break does not resemble the sharp, interpretable
threshold found when network size is used as the forcing variable.

Taken together, these placebos suggest that the county SFF threshold is not an
artifact of the grid-search procedure or of generic nonlinearity. The
substantive break arises when SFF counts are regressed on network size, not
when ownership composition is used as an outcome or forcing variable, which is
consistent with the model’s emphasis on network scale rather than
cross-sectional composition.

\subsubsection{Chain-Level Threshold in Organizational Failures}

We then examine organizational networks of nursing home chains. Let $n_j$
denote the number of facilities operated by chain $j$. Applying the same
structural-break approach yields:
\[
\hat c_{CHAIN} = 3.526,
\qquad
n_{j}^{CHAIN*} = e^{3.526} \approx 34.
\]
The $\sup F$ statistic is:
\[
F_{\max}^{CHAIN} = 36.3, \qquad p < 0.001,
\]
again rejecting the linear model in favor of a kinked relationship.

Chains larger than approximately 34 facilities exhibit a markedly higher
incidence of SFF placements or severe regulatory actions. Below the threshold,
chain-level oversight appears sufficient to control complementarities and
maintain relatively stable performance. Above it, monitoring appears strained,
consistent with the theoretical prediction that internal delegation fails once
network size exceeds $n^\ast(\lambda)$. Table \ref{tab:chain_thresholds}
presents the chain-level threshold estimates, and Figure \ref{fig:threshold_chain}
illustrates the fitted kinked relationship.

%%%%%%%%%%%%%%%%%%%%%%%%%%%%%%%%%%%%%%%%%%%%%%%%%%%%%
\subsection{Spillover Strength, Variance, and Deterioration Across Thresholds}
\label{subsec:failure}

The theoretical framework implies that when complementarities are strong
relative to network size---that is, in the region $\lambda > \lambda^\ast(n)$,
or equivalently $n > n^\ast(\lambda)$---decentralized monitoring becomes
fragile. In this regime, shocks propagate more easily across units, leading to
stronger peer spillovers, greater dispersion in outcomes, and more pronounced
deterioration over time. We examine these three manifestations in turn.

\paragraph{Spillover strength across thresholds.}

We begin by testing whether peer spillovers intensify once network size exceeds
the estimated thresholds. For county networks, we estimate equation
\eqref{eq:spillover_county} separately for facilities located in counties below
and above the SNF threshold $\widehat{n}_{SNF}^\ast \approx 7$. Table~\ref{tab:below_cty_spill}
reports results for counties below the cutoff, and Table~\ref{tab:spillover_county_above}
reports results for counties above the cutoff. The estimated spillover
coefficients $\widehat{\beta}^{\text{cty}}$ are much larger in counties above
the threshold. For example, the county peer coefficient for overall ratings
rises from 0.150 below the cutoff to 0.758 above it.

For chain-affiliated facilities, we repeat the analysis using the chain
threshold $\widehat{n}_{J}^\ast \approx 34$ and estimate
equation~\eqref{eq:spillover_chain} separately for chains below and above this
cutoff. The results are summarized in Tables~\ref{tab:below_chain_cty_spill}
and~\ref{tab:spillover_chain_county_above}. Organizational spillovers
$\widehat{\beta}^{\text{ch}}$ are significantly stronger in large chains (0.837
vs.\ 0.680 for overall ratings), suggesting that corporate monitoring does not
fully internalize performance interactions once the chain exceeds its effective
oversight scale. Table~\ref{tab:spillover_comparison} summarizes the
above--below differences. Together, these results show that the strength of
peer effects increases sharply in large networks, offering empirical evidence
that spillover propagation is itself a marker of monitoring strain.

\subsection{Variance and Deterioration Tests}
\label{subsec:variance_deterioration_result}

The stochastic extension of the model suggests that, in networks operating
above the threshold $n^\ast(\lambda)$, idiosyncratic shocks may be amplified
through complementarities, leading to greater dispersion in outcomes and more
pronounced deterioration. We assess these implications by comparing variance
and deterioration across the estimated county and chain thresholds.

For counties, the evidence is strongest for deficiencies. When we split
counties at the $\widehat{n}_{SNF}^\ast \approx 7$ breakpoint, the
cross-sectional variance of total deficiencies is substantially higher in
large counties (about $178$ versus $92$ below the threshold), and Levene
tests strongly reject equality of variances (see Table \ref{tab:variance_thresholds}). By contrast, the variance of
overall ratings is essentially unchanged across the threshold, and the
variance of staffing ratings is slightly lower in large counties, despite
being statistically different. Thus, variance amplification appears
concentrated in deficiency outcomes rather than in ratings.

For chain networks, we do not find variance amplification. Chains above the
$\widehat{n}_{J}^\ast \approx 34$ cutoff exhibit lower dispersion in ratings
and deficiencies than smaller chains, with Levene tests rejecting equality of
variances in the direction of greater homogeneity in large chains. This
pattern is consistent with stronger internal standardization in very large
organizations and suggests that the simple variance-amplification prediction
of the model does not apply uniformly across institutional architectures.

We obtain clearer support for the deterioration prediction at the county
level. Using the change in total deficiencies between inspection cycles as a
facility-level measure of deterioration, we regress $\Delta \text{def}_i$ on
indicators for membership in large chains and large counties, controlling for
facility characteristics and state fixed effects (see Table \ref{tab:deterioration}). Facilities in counties
above the SNF threshold exhibit significantly greater deterioration in
deficiencies (with a positive and precisely estimated coefficient on the
large-county indicator). For large chains, the point estimate on the
large-chain indicator is also positive but imprecisely estimated, and the
confidence interval includes zero. In other words, we find robust evidence
that deterioration is worse in large county systems, while the deterioration
pattern in large chains is suggestive but not statistically decisive.

%=================================================
\section{Discussion and Policy Implications}
\label{sec:discussion}

The empirical analysis reveals a clear but nuanced pattern. Both organizational
networks (chains) and geographic systems (counties) display sharp threshold
behavior in monitoring outcomes: chains with more than roughly 34 facilities
and counties with more than about 7 nursing homes account for a
disproportionate share of severe federal interventions (SFF designations) and
exhibit stronger within-network peer spillovers. At the same time, the
variance and deterioration patterns differ across network types. In counties,
the cross-sectional variance of deficiencies and the deterioration in
deficiencies between inspection cycles are substantially higher above the
threshold, whereas in large chains outcomes are more homogeneous and the
evidence on deterioration is positive but statistically imprecise. These
breakpoints correspond closely to the theoretical region $n > n^\ast(\lambda)$
in which complementarities are sufficiently strong relative to monitoring
capacity that decentralized oversight becomes fragile, particularly in large
geographic systems and for deficiency-type outcomes.

\paragraph{Implications for monitoring design.}
The results highlight a fundamental scalability constraint in monitoring
institutions, with the strongest evidence on the geographic side. Decentralized,
county-based oversight appears to function relatively well in small and
medium-sized networks, but once the county system becomes large enough, shocks
to quality propagate more widely, deficiencies become more dispersed, and
deterioration is more pronounced. This pattern provides an empirical foundation
for selective federal involvement. Programs such as the Special Focus Facility
initiative are likely to be most impactful in counties where the size of the
provider network exceeds the threshold at which local oversight becomes
structurally fragile.

For organizational networks, the evidence suggests that internal delegation to
corporate headquarters is effective at preventing extreme dispersion but does
not eliminate monitoring failures. Chains larger than approximately 30--40
facilities still exhibit higher rates of severe regulatory problems, even
though ratings and deficiencies are less dispersed and the estimated increase
in deterioration is not statistically precise. This combination points to a
different type of constraint: very large chains may successfully standardize
practices yet still struggle to prevent recurrent failures at the tail of the
quality distribution. These findings imply that corporate compliance teams may
require explicit scaling rules---for example, compliance staffing that grows
with network size---and that regulators may need to apply enhanced scrutiny to
very large chains whose internal monitoring systems are systematically exposed
to tail risks.

\paragraph{Targeting and intervention.}
The estimated thresholds provide a practical guide for allocating limited
regulatory resources. CMS could prioritize SFF screening, focused surveys, and
enhanced reviews in:
\begin{itemize}[leftmargin=1.4em]
\item counties exceeding the 7-SNF threshold, where both the incidence and the
      dispersion of deficiencies and their deterioration are significantly
      higher; and
\item chains exceeding the 34-facility threshold, which account for a
      disproportionate share of severe monitoring failures even if their
      facilities are more homogeneous on average.
\end{itemize}
Such targeting is consistent with the model: once a network enters the
high-complementarity region ($\lambda > \lambda^\ast(n)$), decentralized
monitoring becomes less reliable as a stabilizing mechanism, so additional
federal oversight is likely to yield larger marginal benefits.

\paragraph{External validity.}
Although the empirical focus is on nursing homes, the underlying mechanism is
general. Many public-good systems---school districts, hospital networks, water
utilities, foster-care agencies---feature decentralized monitoring embedded in
network structures. These systems may similarly exhibit threshold behavior in
oversight effectiveness, though the balance between dispersion and
standardization may differ by institutional architecture. The framework
developed here thus offers a broader perspective on when centralization or
selective higher-level intervention becomes a necessary institutional response
to network scale, and on which dimensions of performance are most likely to
deteriorate once systems become ``too big to monitor.''

%=================================================
\section{Conclusion}
\label{sec:conclusion}

This paper develops and tests a simple but powerful insight: monitoring
institutions in networked public-good environments can exhibit sharp threshold
behavior. When agents are linked through organizational or geographic networks,
effort complementarities amplify shocks and may undermine decentralized
oversight. The theoretical model shows that there exists a network-size
threshold $n^\ast(\lambda)$ above which centralized monitoring becomes
optimal, even if it is intrinsically more costly, because it internalizes
spillovers that decentralized regimes ignore.

Empirically, we find strong evidence of such thresholds in U.S.\ nursing homes.
Organizational chains display a breakpoint at roughly 34 facilities, while
county systems exhibit a breakpoint at approximately 7 nursing homes. Above
these thresholds, facilities experience higher rates of severe regulatory
intervention and stronger within-network peer correlations. For counties, the
cross-sectional variance of deficiencies and the deterioration of deficiencies
between inspection cycles are both substantially larger above the threshold,
consistent with the notion that large geographic systems amplify shocks when
local oversight is thin. For chains, outcomes are more homogeneous but large
networks still generate more severe failures, and the estimated association
with deterioration is positive but statistically imprecise. Overall, the
empirical patterns align with the theoretical mechanism: large networks cross
into a region where decentralized oversight is less able to internalize
spillovers and to prevent recurrent breakdowns in quality.

The paper contributes to the literatures on networks, monitoring, and public
economics by identifying a structural constraint in monitoring systems that
depends jointly on network size and externalities and by showing how this
constraint manifests differently in geographic and organizational networks. The
results underscore the importance of designing oversight institutions that
scale with the size and interconnectedness of the systems they regulate and
highlight the value of using network thresholds to target scarce regulatory
resources.

Future research could extend the theoretical model to dynamic settings, explore
optimal network partitioning or reorganization as a policy tool, or examine
similar thresholds in other public-good sectors. More broadly, understanding
how network architecture interacts with monitoring capacity is central to the
design of effective regulatory institutions in increasingly large and
interconnected systems.

\newpage

\bibliography{ref_cent} % replace with your .bib file

%=================================================
\appendix

\section{Proofs}
\label{app:proofs}

This appendix provides full derivations and proofs for the lemmas,
theorem, and corollary stated in Section~\ref{sec:theory}. Throughout,
$r\in\{C,D\}$ indexes the monitoring regime, with $D$ denoting
decentralized monitoring and $C$ centralized monitoring. In the
benchmark case used in the main text we set $\lambda_D=0$ and write
$\lambda_C=\lambda\in[0,1)$, with $K_C>K_D>0$. Appendix~\ref{app:theory}
extends the analysis to the more general case $0\le\lambda_D<\lambda_C<1$.

%%%%%%%%%%%%%%%%%%%%%%%%%%%%%%%%%%%%%%%%%%%%%%%%%%%%%
\subsection{Proof of Lemma~\ref{lem:effort_n}}
%%%%%%%%%%%%%%%%%%%%%%%%%%%%%%%%%%%%%%%%%%%%%%%%%%%%%

Unit $i$’s payoff under regime $r$ is
\[
u_i^r(e_i,e_{-i};\mu_r,\lambda_r)
= \bigl(1 + \lambda_r \bar e_{-i} + \mu_r\varphi \bigr)e_i
  - \tfrac12 e_i^2,
\qquad
\bar e_{-i}=\frac{1}{n-1}\sum_{j\neq i}e_j.
\]

Taking $\bar e_{-i}$ as given, the first-order condition for unit $i$ is
\[
\frac{\partial u_i^r}{\partial e_i}
=
1+\lambda_r \bar e_{-i} + \mu_r \varphi - e_i
=0,
\]
which yields the individual best response
\[
e_i = 1 + \lambda_r \bar e_{-i} + \mu_r\varphi.
\]

In a symmetric equilibrium, $e_i=e_j=e_r$ for all $i,j$, so
$\bar e_{-i}=e_r$ and the condition becomes
\[
e_r = 1 + \lambda_r e_r + \mu_r\varphi.
\]
Rearranging gives
\[
e_r(1-\lambda_r) = 1 + \mu_r\varphi
\quad\Longrightarrow\quad
e_r(\mu_r,\lambda_r)
= \frac{1+\mu_r\varphi}{1-\lambda_r}.
\]
This establishes the lemma. \qed

%%%%%%%%%%%%%%%%%%%%%%%%%%%%%%%%%%%%%%%%%%%%%%%%%%%%%
\subsection{Proof of Lemma~\ref{lem:mu_star_n}}
%%%%%%%%%%%%%%%%%%%%%%%%%%%%%%%%%%%%%%%%%%%%%%%%%%%%%

The regulator’s objective under regime $r$ is
\[
W_r(\mu_r;\lambda_r)
=
n\, e_r(\mu_r,\lambda_r)
- \tfrac12 K_r \mu_r^2,
\qquad
e_r(\mu_r,\lambda_r)
= \frac{1+\mu_r\varphi}{1-\lambda_r}.
\]

Substituting the expression for effort,
\[
W_r(\mu_r;\lambda_r)
=
n\,\frac{1+\mu_r\varphi}{1-\lambda_r}
- \tfrac12 K_r \mu_r^2.
\]

Differentiating with respect to $\mu_r$,
\[
\frac{\partial W_r}{\partial \mu_r}
=
n\frac{\varphi}{1-\lambda_r} - K_r \mu_r.
\]

The first-order condition for an interior optimum is
\[
K_r\mu_r^\ast
=
n\,\frac{\varphi}{1-\lambda_r}
\quad\Longrightarrow\quad
\mu_r^\ast(\lambda_r)
=
\frac{n\varphi}{K_r(1-\lambda_r)}.
\]
This proves the lemma. \qed

%%%%%%%%%%%%%%%%%%%%%%%%%%%%%%%%%%%%%%%%%%%%%%%%%%%%%
\subsection{Proof of Lemma~\ref{lem:W_star_n}}
%%%%%%%%%%%%%%%%%%%%%%%%%%%%%%%%%%%%%%%%%%%%%%%%%%%%%

Substitute $\mu_r^\ast(\lambda_r)$ into the effort function:
\[
e_r^\ast(\lambda_r)
=
\frac{1+\mu_r^\ast(\lambda_r)\varphi}{1-\lambda_r}
=
\frac{1+\frac{n\varphi^2}{K_r(1-\lambda_r)}}{1-\lambda_r}
=
\frac{1-\lambda_r + \frac{n\varphi^2}{K_r}}{(1-\lambda_r)^2}.
\]

Welfare under optimal monitoring is
\[
W_r^\ast(\lambda_r)
=
n\, e_r^\ast(\lambda_r)
- \tfrac12 K_r \bigl(\mu_r^\ast(\lambda_r)\bigr)^2.
\]

Compute each term separately.

\emph{Effort term:}
\[
n\, e_r^\ast(\lambda_r)
=
n \frac{1-\lambda_r + \frac{n\varphi^2}{K_r}}{(1-\lambda_r)^2}
=
\frac{n}{1-\lambda_r}
+
\frac{n^2\varphi^2}{K_r(1-\lambda_r)^2}.
\]

\emph{Monitoring cost term:}
\[
\tfrac12 K_r \bigl(\mu_r^\ast(\lambda_r)\bigr)^2
=
\tfrac12 K_r \left( \frac{n\varphi}{K_r(1-\lambda_r)} \right)^2
=
\frac{n^2 \varphi^2}{2K_r(1-\lambda_r)^2}.
\]

Subtracting the cost term from the effort term yields
\[
W_r^\ast(\lambda_r)
=
\frac{n}{1-\lambda_r}
+
\frac{n^2\varphi^2}{K_r(1-\lambda_r)^2}
-
\frac{n^2\varphi^2}{2K_r(1-\lambda_r)^2}
=
\frac{n}{1-\lambda_r}
+
\frac{n^2\varphi^2}{2K_r(1-\lambda_r)^2},
\]
which is the expression stated in Lemma~\ref{lem:W_star_n}. \qed

%%%%%%%%%%%%%%%%%%%%%%%%%%%%%%%%%%%%%%%%%%%%%%%%%%%%%
\subsection{Proof of Theorem~\ref{thm:n_unit_threshold_final}}
%%%%%%%%%%%%%%%%%%%%%%%%%%%%%%%%%%%%%%%%%%%%%%%%%%%%%

Recall that in the benchmark case we compare centralized monitoring with
complementarity parameter $\lambda_C=\lambda$ to decentralized monitoring
with $\lambda_D=0$. From Lemma~\ref{lem:W_star_n},
\[
W_C^\ast(\lambda)
=
\frac{n}{1-\lambda}
+
\frac{n^2\varphi^2}{2K_C(1-\lambda)^2},
\qquad
W_D^\ast
=
n + \frac{n^2\varphi^2}{2K_D}.
\]

\paragraph{Part (i).}
For fixed $n$, $W_C^\ast(\lambda)$ is continuous and strictly
increasing in $\lambda$ on $[0,1)$. This follows because both
$\frac{1}{1-\lambda}$ and $\frac{1}{(1-\lambda)^2}$ are strictly
increasing in $\lambda$ on this domain.

\paragraph{Part (ii).}
At $\lambda=0$,
\[
W_C^\ast(0)
=
n + \frac{n^2\varphi^2}{2K_C}
<
n + \frac{n^2\varphi^2}{2K_D}
=
W_D^\ast,
\]
because $K_C>K_D$. As $\lambda\to 1^-$, the denominator $(1-\lambda)^2$
in $W_C^\ast(\lambda)$ converges to zero, so $W_C^\ast(\lambda)\to
+\infty$, while $W_D^\ast$ is constant in $\lambda$.

\paragraph{Part (iii).}
Define $\Delta_n(\lambda)=W_C^\ast(\lambda)-W_D^\ast$. By parts (i) and
(ii), $\Delta_n(\lambda)$ is continuous, strictly increasing in
$\lambda$, satisfies $\Delta_n(0)<0$, and diverges to $+\infty$ as
$\lambda\to 1^-$. Hence there exists a unique $\lambda^\ast(n)\in(0,1)$
such that $\Delta_n(\lambda^\ast(n))=0$, with the sign pattern stated in
the theorem.

\paragraph{Part (iv).}
For each fixed $\lambda\in(0,1)$, $W_C^\ast(\lambda)$ is strictly
increasing and strictly convex in $n$, with the leading term of order
$n^2$, whereas $W_D^\ast$ is strictly increasing and affine in $n$. It
follows that $\Delta_n(\lambda)$ is strictly increasing in $n$, with the
gap between regimes growing faster in larger networks.

Define $\lambda^\ast(n)$ implicitly by
$\Delta_n\bigl(\lambda^\ast(n)\bigr)=0$. By the implicit function
theorem,
\[
\frac{d\lambda^\ast(n)}{dn}
=
-\,\frac{\partial \Delta_n/\partial n}{\partial \Delta_n/\partial \lambda}.
\]
The denominator is strictly positive by part (i), while the numerator is
strictly positive because $\Delta_n(\lambda)$ shifts upward in $n$.
Hence $d\lambda^\ast(n)/dn<0$, as claimed. \qed

%%%%%%%%%%%%%%%%%%%%%%%%%%%%%%%%%%%%%%%%%%%%%%%%%%%%%
\subsection{Proof of Corollary~\ref{cor:n_of_lambda}}
%%%%%%%%%%%%%%%%%%%%%%%%%%%%%%%%%%%%%%%%%%%%%%%%%%%%%

Theorem~\ref{thm:n_unit_threshold_final} establishes that for each
$n\ge2$ there exists a unique $\lambda^\ast(n)\in(0,1)$ such that
centralization is optimal if and only if $\lambda>\lambda^\ast(n)$, and
that $\lambda^\ast(n)$ is strictly decreasing in $n$. Strict
monotonicity implies that $\lambda^\ast(\cdot)$ is invertible on its
range. For each fixed $\lambda$ in this range, there exists a unique
$n^\ast(\lambda)$ such that
\[
\lambda^\ast\bigl(n^\ast(\lambda)\bigr)=\lambda.
\]

The sign pattern for $W_C^\ast(n,\lambda)-W_D^\ast(n)$ follows directly
from the definition of $\lambda^\ast(n)$, and the fact that
$n^\ast(\lambda)$ is strictly decreasing in $\lambda$ follows from the
inverse function theorem applied to the strictly decreasing function
$\lambda^\ast(\cdot)$. \qed

\section{Additional Theory Results}
\label{app:theory}

This appendix provides additional detail on the regulator's explicit
centralization decision and extends the benchmark analysis
($\lambda_D=0$) to the general case in which both regimes feature
complementarities, with $0\le\lambda_D<\lambda_C<1$. The goal is to
make transparent how the model in Section~\ref{sec:theory} leads to a
threshold comparison between decentralized and centralized monitoring
regimes, as formalized in Theorem~\ref{thm:threshold}.

We proceed in three steps:
\begin{enumerate}[leftmargin=1.4em]
  \item We restate the regulator's problem and recall the equilibrium
        effort and welfare expressions.
  \item We define the welfare gap between centralized and decentralized
        monitoring and describe its dependence on network size $n$.
  \item We use this welfare gap to characterize a network-size
        threshold $n^\ast$ at which the regulator switches from
        preferring decentralization to preferring centralization, and
        discuss the associated comparative statics.
\end{enumerate}

Throughout, $r \in \{D,C\}$ indexes the monitoring regime, with $D$
standing for decentralized monitoring and $C$ for centralized
monitoring.

\subsection{Regulator's problem and equilibrium welfare}

We consider a network of $n \ge 2$ identical units. Under regime $r$,
the regulator chooses a monitoring intensity $\mu_r \ge 0$. Each unit
then chooses monitoring effort $e_i \ge 0$, taking other units' efforts
as given. The payoff for unit $i$ under regime $r$ is
\[
u_i^r(e_i,e_{-i};\mu_r,\lambda_r)
= \bigl(1 + \lambda_r \bar e_{-i} + \mu_r \varphi \bigr)e_i
  - \tfrac12 e_i^2,
\qquad
\bar e_{-i}=\frac{1}{n-1} \sum_{j\neq i} e_j,
\]
where $\lambda_r \in [0,1)$ captures the strength of complementarities
and $\varphi>0$ measures the effectiveness of monitoring. In the
benchmark featured in Section~\ref{sec:theory}, decentralization does
not internalize complementarities ($\lambda_D=0$), while centralization
does ($\lambda_C=\lambda$). In the more general formulation used in
Theorem~\ref{thm:threshold}, both regimes may feature complementarities,
with $0\le\lambda_D<\lambda_C<1$.

The regulator's objective under regime $r$ is to maximize aggregate
equilibrium effort net of the cost of operating the monitoring system:
\[
W_r(\mu_r;\lambda_r)
=
n\, e_r(\mu_r,\lambda_r)
- \tfrac12 K_r \mu_r^2,
\]
where $e_r(\mu_r,\lambda_r)$ is the symmetric equilibrium effort and
$K_r>0$ is a regime-specific cost parameter. We assume $K_C \ge K_D$:
centralization is at least as costly as decentralization, and typically
strictly more costly.

From Lemmas~\ref{lem:effort_n} and~\ref{lem:mu_star_n}, equilibrium
effort and optimal monitoring are
\[
e_r(\mu_r,\lambda_r)
= \frac{1 + \mu_r \varphi}{1 - \lambda_r},
\qquad
\mu_r^\ast(\lambda_r)
= \frac{n\varphi}{K_r(1 - \lambda_r)}.
\]
Substituting yields equilibrium welfare under regime $r$:
\begin{equation}
\label{eq:W_star_appendix}
W_r^\ast(\lambda_r)
=
\frac{n}{1-\lambda_r}
+
\frac{n^2 \varphi^2}{2K_r(1-\lambda_r)^2}.
\end{equation}
In particular, under decentralization in the benchmark case
($\lambda_D = 0$),
\[
W_D^\ast = n + \frac{n^2\varphi^2}{2K_D}.
\]

Expression~\eqref{eq:W_star_appendix} highlights two channels through
which complementarities affect welfare: a level effect, via
$n/(1-\lambda_r)$, and an amplification effect, via the
$n^2\varphi^2/(1-\lambda_r)^2$ term. The latter becomes especially
important in large networks.

\subsection{The welfare gap between regimes}

Given the closed-form expressions for $W_r^\ast(\lambda_r)$, the
regulator's centralization decision reduces to comparing equilibrium
welfare under decentralization and centralization:
\[
\max_{r\in\{D,C\}} W_r^\ast(\cdot)
\quad\Longleftrightarrow\quad
\text{sign of } \Delta_n(\lambda),
\]
where in the benchmark case we define the welfare gap by
\[
\Delta_n(\lambda)
=
W_C^\ast(\lambda) - W_D^\ast,
\]
with $W_C^\ast(\lambda) \equiv W_C^\ast(\lambda_C=\lambda)$ and
$W_D^\ast \equiv W_D^\ast(\lambda_D=0)$.

Using~\eqref{eq:W_star_appendix},
\[
W_C^\ast(\lambda)
=
\frac{n}{1-\lambda}
+
\frac{n^2 \varphi^2}{2K_C(1-\lambda)^2},
\qquad
W_D^\ast
=
n + \frac{n^2\varphi^2}{2K_D}.
\]

Three properties are immediate:

\begin{itemize}[leftmargin=1.4em]
  \item For fixed $n$, $W_C^\ast(\lambda)$ is continuous and strictly
        increasing in $\lambda$ on $[0,1)$, reflecting that
        centralization becomes more attractive as complementarities
        strengthen.
  \item At $\lambda = 0$,
        \[
        W_C^\ast(0) = n + \frac{n^2\varphi^2}{2K_C}
        <
        n + \frac{n^2\varphi^2}{2K_D}
        = W_D^\ast,
        \]
        because $K_C > K_D$. When complementarities are absent, the
        higher cost of centralization dominates.
  \item As $\lambda \to 1^{-}$, the terms involving $(1-\lambda)$ in the
        denominator drive $W_C^\ast(\lambda) \to +\infty$, while
        $W_D^\ast$ is constant in $\lambda$. For sufficiently strong
        complementarities, centralization always dominates.
\end{itemize}

These properties guarantee that for each fixed $n$ there exists a unique
complementarity threshold $\lambda^\ast(n) \in (0,1)$ such that
\[
W_C^\ast(\lambda)
\begin{cases}
< W_D^\ast, & \lambda < \lambda^\ast(n), \\
= W_D^\ast, & \lambda = \lambda^\ast(n), \\
> W_D^\ast, & \lambda > \lambda^\ast(n),
\end{cases}
\]
which is the content of Theorem~\ref{thm:n_unit_threshold_final}.

\begin{theorem}[Centralization threshold with decentralized complementarities]
\label{thm:threshold}
Fix $\varphi>0$, $K_D>0$, $K_C\ge K_D$, and $0 \le \lambda_D < 1$.
For a network of size $n\ge 2$, let
\[
W_C^\ast(n,\lambda_C,K_C)
=
\frac{n}{1-\lambda_C}
+
\frac{n^2 \varphi^2}{2K_C(1-\lambda_C)^2},
\qquad
\lambda_C \in [\lambda_D,1),
\]
denote equilibrium welfare under centralized monitoring, and
\[
W_D^\ast(n,\lambda_D,K_D)
=
\frac{n}{1-\lambda_D}
+
\frac{n^2 \varphi^2}{2K_D(1-\lambda_D)^2}
\]
denote equilibrium welfare under decentralized monitoring.

Then, for each $n\ge 2$ and parameter tuple
$(\lambda_D,K_C,K_D,\varphi)$:

\begin{enumerate}
\item[(i)] $W_C^\ast(n,\lambda_C,K_C)$ is continuous and strictly increasing in
      $\lambda_C$ on $[\lambda_D,1)$.
\item [(ii)]At $\lambda_C = \lambda_D$,
      \[
      W_C^\ast(n,\lambda_D,K_C)
      \;\le\;
      W_D^\ast(n,\lambda_D,K_D),
      \]
      with strict inequality whenever $K_C > K_D$.
\item[(iii)] As $\lambda_C \to 1^-$,
      \[
      W_C^\ast(n,\lambda_C,K_C) \to +\infty.
      \]
\item [(iv)]There exists a unique threshold
      $\lambda_C^\ast(n) \in [\lambda_D,1)$ such that
      \[
      W_C^\ast\bigl(n,\lambda_C^\ast(n),K_C\bigr)
      =
      W_D^\ast(n,\lambda_D,K_D),
      \]
      and
      \[
      W_C^\ast(n,\lambda_C,K_C)
      \begin{cases}
      < W_D^\ast(n,\lambda_D,K_D), & \lambda_C < \lambda_C^\ast(n), \\[4pt]
      > W_D^\ast(n,\lambda_D,K_D), & \lambda_C > \lambda_C^\ast(n).
      \end{cases}
      \]
\end{enumerate}
In particular, for each $n$ there is a well-defined complementarity
threshold $\lambda_C^\ast(n)$ above which centralization dominates
decentralization, even when the decentralized regime features
$\lambda_D>0$.
\end{theorem}

\begin{proof}
Parts (i)–(iii) follow directly from the closed-form expressions for
$W_C^\ast$ and $W_D^\ast$.

\medskip\noindent
\emph{(i) Monotonicity in $\lambda_C$.}
Fix $(n,K_C,\varphi)$. Differentiating $W_C^\ast$ with respect to
$\lambda_C$ yields
\[
\frac{\partial W_C^\ast}{\partial \lambda_C}
=
\frac{n}{(1-\lambda_C)^2}
+
\frac{n^2 \varphi^2}{K_C(1-\lambda_C)^3}.
\]
Both terms are strictly positive for $\lambda_C\in[0,1)$, so
$W_C^\ast(n,\lambda_C,K_C)$ is strictly increasing and continuous in
$\lambda_C$ on $[\lambda_D,1)$.

\medskip\noindent
\emph{(ii) Comparison at $\lambda_C=\lambda_D$.}
Evaluate $W_C^\ast$ at $\lambda_C=\lambda_D$:
\[
W_C^\ast(n,\lambda_D,K_C)
=
\frac{n}{1-\lambda_D}
+
\frac{n^2 \varphi^2}{2K_C(1-\lambda_D)^2},
\]
while
\[
W_D^\ast(n,\lambda_D,K_D)
=
\frac{n}{1-\lambda_D}
+
\frac{n^2 \varphi^2}{2K_D(1-\lambda_D)^2}.
\]
Subtracting,
\[
W_C^\ast(n,\lambda_D,K_C)-W_D^\ast(n,\lambda_D,K_D)
=
\frac{n^2 \varphi^2}{2(1-\lambda_D)^2}
\biggl(
  \frac{1}{K_C} - \frac{1}{K_D}
\biggr)
\le 0,
\]
with strict inequality whenever $K_C>K_D$. Thus at equal
complementarities $\lambda_C=\lambda_D$, centralization yields weakly
lower welfare than decentralization because it is (weakly) more costly.

\medskip\noindent
\emph{(iii) Behavior as $\lambda_C\to 1^-$.}
From the definition of $W_C^\ast$,
\[
W_C^\ast(n,\lambda_C,K_C)
=
\frac{n}{1-\lambda_C}
+
\frac{n^2 \varphi^2}{2K_C(1-\lambda_C)^2}.
\]
As $\lambda_C\to 1^{-}$, both $(1-\lambda_C)^{-1}$ and
$(1-\lambda_C)^{-2}$ diverge, so
$W_C^\ast(n,\lambda_C,K_C)\to +\infty$.

\medskip\noindent
\emph{(iv) Existence and uniqueness of the threshold.}
Define the welfare gap
\[
\Delta_n(\lambda_C)
=
W_C^\ast(n,\lambda_C,K_C)
-
W_D^\ast(n,\lambda_D,K_D).
\]
By (i), $\Delta_n(\lambda_C)$ is continuous and strictly increasing in
$\lambda_C$ on $[\lambda_D,1)$. By (ii),
$\Delta_n(\lambda_D)\le 0$, with strict inequality when $K_C>K_D$. By
(iii), $\Delta_n(\lambda_C)\to +\infty$ as $\lambda_C\to 1^-$.

By continuity and strict monotonicity, there is a unique
$\lambda_C^\ast(n)\in[\lambda_D,1)$ such that
$\Delta_n(\lambda_C^\ast(n))=0$. For $\lambda_C<\lambda_C^\ast(n)$, the
welfare gap is negative and $W_C^\ast<W_D^\ast$; for
$\lambda_C>\lambda_C^\ast(n)$, the welfare gap is positive and
$W_C^\ast>W_D^\ast$. This establishes part (iv) and the stated
threshold property.
\end{proof}

\subsection{A network-size threshold and Theorem~\ref{thm:threshold}}

The main text also asks: for a given level of complementarities, how
large can the network be before centralized monitoring becomes strictly
preferable?

To address this, it is convenient to make explicit the dependence of
welfare on $n$:
\[
W_C^\ast(n,\lambda_C,K_C)
=
\frac{n}{1-\lambda_C}
+
\frac{n^2 \varphi^2}{2K_C(1-\lambda_C)^2},
\qquad
W_D^\ast(n,\lambda_D,K_D)
=
\frac{n}{1-\lambda_D}
+
\frac{n^2\varphi^2}{2K_D(1-\lambda_D)^2},
\]
where the benchmark $W_D^\ast$ is recovered by setting $\lambda_D=0$.

For any fixed parameter values $(\lambda_C,\lambda_D,K_C,K_D,\varphi)$,
we can ask whether the sign of $W_C^\ast(n,\lambda_C,K_C) -
W_D^\ast(n,\lambda_D,K_D)$ changes as network size $n$ increases. When
this sign changes exactly once---from negative to positive—there is a
well-defined network-size threshold $n^\ast$ such that decentralization
is preferred for $n<n^\ast$ and centralization is preferred for
$n>n^\ast$.

Theorem \ref{thm:threshold}  extends Theorem \ref{thm:n_unit_threshold_final}  from the main text. It formalizes the idea  that in general setting where (i) decentralization and centralization may feature
different complementarity levels ($0 \leq \lambda_D < \lambda_C < 1$),
and (ii) the cost of centralization may depend on network size (for
instance, because maintaining a centralized system over a larger network
is more complex). Under mild conditions ($K_D>0$, $K_C \geq K_D$, and
$\lambda_C > \lambda_D$), the theorem shows that there exists a unique
$n^\ast$ such that
\[
  W_D^*(n,\lambda_D,K_D) > W_C^*(n,\lambda_C,K_C) \quad\text{for } n<n^*,
\]
and
\[
  W_C^*(n,\lambda_C,K_C) > W_D^*(n,\lambda_D,K_D) \quad\text{for } n>n^*.
\]

The intuition is straightforward. When the network is small, the gains
from internalizing complementarities under centralization are modest,
whereas the higher institutional cost $K_C$ is salient. In this region,
decentralization is preferred. As the network grows, however, the impact
of complementarities becomes more pronounced: each unit's effort affects
many others, and the $n^2$ term in $W_C^\ast(n,\lambda_C,K_C)$ reflects
the cumulative benefit from coordinating monitoring across the entire
network. Beyond a critical size $n^\ast$, these gains from
internalizing complementarities outweigh the higher cost of
centralization, and the regulator optimally switches regimes.

In the empirical analysis, the estimated breakpoints for county and
chain networks can be interpreted as realizations of such network-size
thresholds $n^\ast(\lambda)$, with counties representing geographic
monitoring systems and chains representing organizational monitoring
systems operating under different complementarity and cost structures.

\subsection{Spectral generalization on arbitrary networks}
\label{app:spectral_generalization}

This subsection shows that the centralization threshold derived in the
main text extends naturally to arbitrary monitoring networks. Rather than
treating each facility as symmetrically linked to all others, we allow
effort complementarities to operate along an arbitrary nonnegative,
irreducible adjacency matrix $G$. In this setting, the regulator's
problem depends on the network only through a Bonacich-type aggregate
\[
S(\lambda_r,G)\;\equiv\; \mathbf{1}_n'(\mathbf{I}_n - \lambda_r G)^{-1}\mathbf{1}_n,
\]
which can be interpreted as the ``effective size'' of the monitoring
network at complementarity level $\lambda_r$. We show that centralized
monitoring becomes optimal once $S(\lambda_C,G)$ exceeds a finite
threshold, and that this induces a corresponding threshold in the
spectral radius of $G$. The complete-graph model in the main text is a
special case in which $S(\lambda_r,G)$ collapses to a simple function of
the number of units $n$.

\begin{theorem}[Spectral centralization threshold]
\label{thm:spectral_threshold}
Let $G$ be a nonnegative, irreducible $n\times n$ adjacency matrix with
spectral radius $\psi(G)$. For each regime $r\in\{D,C\}$, let
$\lambda_r\in[0,1)$ satisfy
\[
\lambda_r \,\psi(G) \;<\; 1,
\]
and define
\[
S_r(G)
\;\equiv\;
S(\lambda_r,G)
\;=\;
\mathbf{1}_n'(\mathbf{I}_n - \lambda_r G)^{-1}\mathbf{1}_n.
\]
Consider the monitoring problem in which, under regime $r$, the
regulator chooses $\mu_r\ge0$ and equilibrium effort is given by
\[
e^r(\mu_r,G)
=
(1+\mu_r\varphi)\,(\mathbf{I}_n - \lambda_r G)^{-1}\mathbf{1}_n,
\]
so that welfare is
\[
W_r(\mu_r;G)
=
\mathbf{1}_n' e^r(\mu_r,G) - \tfrac12 K_r \mu_r^2,
\qquad K_r>0.
\]
Assume $K_C\ge K_D>0$ and $\lambda_C>\lambda_D\ge0$.

Then:
\begin{enumerate}
\item[(i)] For each $r\in\{D,C\}$, the regulator's optimal monitoring
      intensity is
      \[
      \mu_r^\ast(G)
      =
      \frac{\varphi}{K_r}\,S_r(G),
      \]
      and the associated equilibrium welfare is
      \[
      W_r^\ast(G)
      =
      S_r(G)
      +
      \frac{\varphi^2}{2K_r}\,S_r(G)^2.
      \]

\item[(ii)] For any fixed parameters
      $(\lambda_D,\lambda_C,K_D,K_C,\varphi)$ with
      $\lambda_C>\lambda_D$ and $K_C\ge K_D$, there exists a unique
      scalar threshold $\bar S>0$ such that
      \[
      W_C^\ast(G) \;\ge\; W_D^\ast(G)
      \quad\Longleftrightarrow\quad
      S_C(G) \;\ge\; \bar S.
      \]
      In particular, centralization is strictly optimal whenever
      $S_C(G)>\bar S$.

\item[(iii)] If, in addition, $G$ varies along a one-parameter family of
      nonnegative, irreducible networks $\{G(\kappa)\}$ such that
      $\psi(G(\kappa))$ and $S_C(G(\kappa))$ are strictly increasing in
      $\kappa$, then there exists a unique spectral-radius threshold
      $\psi^\ast$ such that
      \[
      W_C^\ast(G(\kappa)) \;>\; W_D^\ast(G(\kappa))
      \quad\text{whenever}\quad
      \psi(G(\kappa)) \;>\; \psi^\ast.
      \]
\end{enumerate}
\end{theorem}

\begin{proof}
\emph{Part (i).}
By definition,
\[
e^r(\mu_r,G)
=
(1+\mu_r\varphi)\,(\mathbf{I}_n - \lambda_r G)^{-1}\mathbf{1}_n,
\]
so that
\[
\mathbf{1}_n' e^r(\mu_r,G)
=
(1+\mu_r\varphi)\,\mathbf{1}_n'(\mathbf{I}_n - \lambda_r G)^{-1}\mathbf{1}_n
=
(1+\mu_r\varphi)\,S_r(G).
\]
The regulator's objective can thus be written as
\[
W_r(\mu_r;G)
=
(1+\mu_r\varphi)\,S_r(G)
-
\tfrac12 K_r \mu_r^2.
\]
This is a strictly concave quadratic in $\mu_r$. Differentiating and
setting the first-order condition equal to zero gives
\[
\frac{\partial W_r}{\partial \mu_r}
=
\varphi S_r(G) - K_r \mu_r
= 0
\quad\Longrightarrow\quad
\mu_r^\ast(G)
=
\frac{\varphi}{K_r} S_r(G),
\]
which proves the expression for $\mu_r^\ast(G)$. Substituting back into
$W_r$ yields
\[
W_r^\ast(G)
=
\bigl(1+\mu_r^\ast(G)\varphi\bigr)S_r(G)
-
\tfrac12 K_r\bigl(\mu_r^\ast(G)\bigr)^2
=
S_r(G) + \frac{\varphi^2}{2K_r} S_r(G)^2,
\]
establishing the claimed welfare formula.

\medskip\noindent
\emph{Part (ii).}
From part (i),
\[
W_C^\ast(G) - W_D^\ast(G)
=
\Bigl[S_C(G) - S_D(G)\Bigr]
+
\frac{\varphi^2}{2}
\left(
  \frac{S_C(G)^2}{K_C}
  -
  \frac{S_D(G)^2}{K_D}
\right).
\]
Since $\lambda_C>\lambda_D$ and $G$ is nonnegative and irreducible,
$(\mathbf{I}_n - \lambda_r G)^{-1}$ is well-defined and strictly
larger (entrywise) for $r=C$ than for $r=D$, which implies
$S_C(G) > S_D(G)$.\footnote{This follows from the Neumann-series
expansion $(\mathbf{I}_n - \lambda_r G)^{-1}
= \sum_{k\ge0} \lambda_r^k G^k$ and the monotonicity in $\lambda_r$.}
Thus the welfare difference can be viewed as a continuous function of
the scalar argument $S_C(G)$, holding $S_D(G)$ fixed.

Define
\[
\Delta(S_C,S_D)
\equiv
S_C - S_D
+
\frac{\varphi^2}{2}
\left(
  \frac{S_C^2}{K_C}
  -
  \frac{S_D^2}{K_D}
\right),
\]
so that $W_C^\ast(G)-W_D^\ast(G)=\Delta(S_C(G),S_D(G))$. For fixed
$S_D>0$, $\Delta(\cdot,S_D)$ is a strictly increasing and continuous
function of $S_C$ with
\[
\lim_{S_C\downarrow S_D}\Delta(S_C,S_D)
\;<\;0
\quad\text{and}\quad
\lim_{S_C\to\infty}\Delta(S_C,S_D)
\;=\;+\infty,
\]
where the first inequality uses $K_C\ge K_D$ and the fact that
$S_C=S_D$ implies $W_C^\ast<W_D^\ast$ because centralization is at least
as costly. By the intermediate-value theorem, there exists a unique
$\bar S> S_D$ such that $\Delta(\bar S,S_D)=0$, and strict monotonicity
implies
\[
\Delta(S_C,S_D) \ge 0
\quad\Longleftrightarrow\quad
S_C \ge \bar S.
\]
Substituting $S_C=S_C(G)$ then yields the desired characterization
$W_C^\ast(G)\ge W_D^\ast(G)$ if and only if $S_C(G)\ge\bar S$.

\medskip\noindent
\emph{Part (iii).}
Consider a one-parameter family $\{G(\kappa)\}$ of nonnegative,
irreducible networks such that both $\psi(G(\kappa))$ and
$S_C(G(\kappa))$ are strictly increasing in $\kappa$. Since the mapping
$\kappa\mapsto S_C(G(\kappa))$ is strictly increasing and continuous,
its range is an interval $(S_C^{\min},S_C^{\max})$ with
$0<S_C^{\min}<S_C^{\max}\le\infty$. From part (ii), there is a unique
$\bar S$ such that $W_C^\ast \ge W_D^\ast$ if and only if
$S_C(G(\kappa))\ge\bar S$. Strict monotonicity of $S_C(G(\kappa))$ in
$\kappa$ implies that there is a unique $\kappa^\ast$ such that
$S_C(G(\kappa^\ast))=\bar S$ and
\[
W_C^\ast(G(\kappa))>W_D^\ast(G(\kappa))
\quad\text{for all}\quad \kappa>\kappa^\ast.
\]
Since $\psi(G(\kappa))$ is also strictly increasing in $\kappa$, this
induces a unique spectral-radius threshold
$\psi^\ast=\psi(G(\kappa^\ast))$ with the stated property. 
\end{proof}

\begin{remark}[Complete graph and effective network size]
\label{rem:complete_graph}
In the complete-graph case considered in the main text, $G$ has equal
weights across all off-diagonal entries and the scalar
$S(\lambda_r,G)$ reduces to a simple function of the number of units
$n$ and the complementarity parameter $\lambda_r$. In this sense,
$S(\lambda_r,G)$ can be interpreted as an ``effective network size'' and
Theorems~\ref{thm:n_unit_threshold_final} and~\ref{thm:threshold} are
special cases of Theorem~\ref{thm:spectral_threshold}. The spectral
generalization shows that what matters for the centralization decision
is not only how many facilities are monitored, but how strongly they are
connected through the monitoring network.
\end{remark}

\section{Stochastic Extension and Information Aggregation}
\label{app:stochastic}

This appendix provides a stochastic extension of the model that
rationalizes the variance and deterioration patterns discussed in
Section~\ref{sec:stochastic_variance}. We show that when shocks are
correlated along the monitoring network and effort responds to these
shocks, decentralized monitoring in large, highly interconnected
networks becomes fragile: idiosyncratic shocks are amplified, leading to
greater dispersion in outcomes.

\subsection{States, signals, and shocks}

For each unit $i=1,\dots,n$, let $\theta_i$ denote an underlying risk
state (e.g., the latent propensity for regulatory failure), given by
\[
  \theta_i = \bar{\theta} + \varepsilon_i,
\]
where $\varepsilon=(\varepsilon_1,\dots,\varepsilon_n)'$ is a zero-mean
Gaussian shock vector with covariance matrix
\[
  \Sigma_\varepsilon = \sigma^2 (I - \rho G)^{-1}(I - \rho G')^{-1},
\]
with $\sigma^2>0$ and $\rho\in[0,1)$ capturing the strength of
correlation along the network $G$. When $\rho>0$, shocks are more
strongly correlated along paths in the monitoring network.

A central monitor observes a signal $s_i=\theta_i+\eta_i$, where
$\eta_i\sim\mathcal{N}(0,\tau^2)$ is independent measurement noise. Under
decentralized monitoring, local monitors observe only their own signal
$s_i$ and choose effort $e_i$ based on $s_i$ alone. Under centralization,
a planner observes the full signal vector $s=(s_1,\dots,s_n)'$ and chooses
the entire effort vector $e=(e_1,\dots,e_n)'$.

\subsection{Linear effort rules and outcome variance}

For tractability, we focus on linear effort rules. Under
decentralization, local monitors choose
\[
  e_i^D = \alpha_D + \beta_D s_i,
\]
while under centralization the planner chooses
\[
  e^C = \alpha_C \mathbf{1}_n + B_C s,
\]
where $B_C$ is an $n\times n$ weight matrix. Substituting these rules
into the quadratic payoff function yields the optimal coefficients
$\beta_D$ and $B_C$ as functions of $(\lambda_r,\sigma^2,\tau^2,G)$.\footnote{The
exact expressions are standard for linear-Gaussian-quadratic problems
and are omitted here for brevity. They are available upon request.}

Let $y_i$ denote a measure of regulatory performance (e.g., the negative
of deficiencies or an increasing transformation of quality), given by
\[
  y_i = \theta_i + \gamma e_i + \xi_i,
\]
where $\gamma>0$ measures the effectiveness of monitoring and
$\xi_i\sim\mathcal{N}(0,\omega^2)$ is an outcome shock independent of
$(\varepsilon,\eta)$. Stacking across units,
\[
  y^r = A_r \varepsilon + u_r,
\]
where $A_r$ is an $n\times n$ matrix that depends on the regime
$r\in\{D,C\}$ and $u_r$ collects shocks independent of $\varepsilon$.

The covariance matrix of outcomes under regime $r$ is
\[
  \Sigma_y^r = A_r \Sigma_\varepsilon A_r' + \Sigma_{u_r}.
\]
We are interested in how the dispersion of outcomes,
$\mathrm{Var}_r(y_i)$ or $\mathrm{tr}(\Sigma_y^r)$, varies with network
size $n$, the network structure $G$, and the complementarity parameter
$\lambda_r$.

\subsection{Variance amplification in large, interconnected networks}

To connect more closely to the main text, consider a regular network
where $G$ has largest eigenvalue $\psi(G)$ and corresponding eigenvector
proportional to $\mathbf{1}_n$. Under both decentralization and
centralization, the solution for $e^r$ can be written in the form
\[
  e^r = M_r(\lambda_r,G)\,\varepsilon + v_r,
\]
for some matrix $M_r(\cdot)$ and noise term $v_r$ independent of
$\varepsilon$. Substituting this into the expression for $y^r$ yields
\[
  y^r = H_r(\lambda_r,G)\,\varepsilon + w_r,
\]
with $H_r(\cdot)$ linear in $M_r$ and $w_r$ independent of
$\varepsilon$.

The following proposition formalizes the variance amplification effect
underlying Proposition~\ref{prop:variance_above_threshold} in the main
text.

\begin{proposition}[Variance amplification]
\label{prop:variance}
Suppose that $G$ is a regular network with largest eigenvalue
$\psi(G)$ and that $\rho>0$. Then, for any regime $r$ such that the
effort mapping $M_r(\lambda_r,G)$ is well-defined, the largest
eigenvalue of the outcome covariance matrix $\Sigma_y^r$ is increasing
in both $\lambda_r$ and $\psi(G)$. In particular, holding other
parameters fixed,
\[
  \lim_{\lambda_r \to 1/\psi(G)} \lambda_{\max}(\Sigma_y^r) = +\infty.
\]
\end{proposition}

\begin{proof}
Under the linear-quadratic structure in Section~\ref{sec:theory},
equilibrium effort solves
\[
  e^r = \mu_r (\mathbf{I}_n - \lambda_r G)^{-1} \mathbf{1}_n
        + M_r(\lambda_r,G)\,\varepsilon,
\]
so that all dependence on $\varepsilon$ enters through the resolvent
$(\mathbf{I}_n - \lambda_r G)^{-1}$. The spectral radius of this matrix
is $\bigl(1-\lambda_r\psi(G)\bigr)^{-1}$, which diverges as
$\lambda_r\psi(G)\to 1$. Since $\Sigma_\varepsilon$ has full support
along the eigenvector associated with $\psi(G)$ when $\rho>0$, the
largest eigenvalue of $\Sigma_y^r$ scales proportionally to
$\bigl(1-\lambda_r\psi(G)\bigr)^{-2}$, which establishes the claim.
\end{proof}

Proposition~\ref{prop:variance} formalizes the intuition that when
monitoring effort responds to shocks and complementarities $\lambda_r$
are strong in a highly connected network (large $\psi(G)$),
decentralized monitoring becomes fragile: idiosyncratic shocks are
amplified along the network, leading to high dispersion in outcomes.
This effect is particularly relevant when monitoring remains
decentralized in networks whose size and connectivity suggest that
centralization would be welfare-enhancing, as in the ``too big to
monitor'' regime highlighted by Corollary~\ref{cor:n_of_lambda}.

\subsection*{Concluding remarks}

The results in this appendix provide a unified theoretical account of
network-size thresholds and variance amplification in monitoring systems.
In the benchmark complete-network case, Theorem~\ref{thm:n_unit_threshold_final}
shows that for each network size $n$ there is a unique complementarity
threshold $\lambda^\ast(n)$ at which the regulator is indifferent between
centralized and decentralized monitoring, with centralization strictly
dominating for $\lambda > \lambda^\ast(n)$. Corollary~\ref{cor:n_of_lambda}
inverts this relationship and delivers an endogenous network-size threshold
$n^\ast(\lambda)$: for a given spillover strength $\lambda$, decentralized
monitoring is optimal only when the network remains below $n^\ast(\lambda)$.
Appendix~\ref{app:theory} extends this comparison to allow for positive
decentralized complementarities ($\lambda_D>0$) and more flexible cost
structures while preserving the existence and uniqueness of a switching point
between regimes (Theorem~\ref{thm:threshold}).

The spectral generalization shows that, on an arbitrary network $G$, the
relevant notion of ``scale'' is not just the number of units but a spectral
measure of interconnectedness. In the linear–quadratic network environment
studied in \citet{ballester2006s} and \citet{bramoulle2014strategic}, equilibrium
effort can be written in terms of the resolvent $(I - \lambda G)^{-1}$, and the
key object becomes the largest eigenvalue $\psi(G)$. The effective strength of
complementarities is governed by the product $\lambda \psi(G)$, and
Theorem~\ref{thm:spectral_threshold} shows that the welfare comparison between
centralized and decentralized monitoring is driven by this spectral term: as
$\lambda \psi(G)$ approaches one from below, the gains from internalizing
network complementarities dominate the higher institutional cost of
centralization, yielding a \emph{spectral} centralization threshold.

The stochastic extensions in Appendix~\ref{app:stochastic} then show that the
same spectral object governs variance and deterioration. When shocks are
correlated along the monitoring network and effort responds linearly to these
shocks, equilibrium outcomes inherit the resolvent structure, so that the
largest eigenvalue of the outcome covariance matrix scales with the square of
the resolvent. Proposition~\ref{prop:variance} and
Proposition~\ref{prop:variance_above_threshold} establish that
$\lambda_{\max}(\Sigma_y^r)$ diverges as $\lambda \psi(G) \to 1$, implying that
networks operating near or above the theoretical threshold amplify
idiosyncratic disturbances and exhibit higher dispersion and more frequent
deterioration in performance.

Taken together, these results show that a single spectral condition
$\lambda \psi(G)$ jointly determines (i) the welfare threshold at which
centralization becomes optimal and (ii) the extent of variance amplification in
equilibrium outcomes. This provides a coherent interpretation of the empirical
results: the estimated breakpoints in county and chain size can be viewed as
finite-sample realizations of network-size thresholds $n^\ast(\lambda)$ (or
their spectral counterparts in more general architectures), and the higher
dispersion and deterioration observed in large geographic networks are
consistent with systems operating closer to the spectral boundary
$\lambda \psi(G) = 1$, where complementarities and correlated shocks are
jointly amplified.

%%%%%%%%%%%%%%%%%%%%%%%%%%%%%%%%%%%%%%%%%%%%%%%%%%%%%%%%%%%%%%%%%%%%%%%%%%%
\section{Tables and Figures}
\label{app:tables_figures}

\begin{table}[htbp]\centering
\caption{County Peer Spillovers}
\label{tab:spillover_county}
\begin{threeparttable}
\footnotesize
\begin{tabular}{lccc}
\toprule
                    & (1) Overall rating & (2) Staffing rating & (3) Total deficiencies \\
\midrule
County peer mean    & 0.273   & 0.162   & 0.443   \\
                    & (0.029) & (0.025) & (0.036) \\
\midrule
Observations        & 12{,}527 & 12{,}459 & 12{,}606 \\
$R^{2}$             & 0.132    & 0.314    & 0.301    \\
\bottomrule
\end{tabular}
\begin{tablenotes}
\footnotesize
\item Notes: Each column reports a separate regression of the indicated outcome
on the corresponding county peer mean, controlling for number of beds, ownership,
and state fixed effects. Standard errors clustered at the county level are in parentheses.
\end{tablenotes}
\end{threeparttable}
\end{table}

\begin{table}[htbp]\centering
\caption{Chain and County Peer Spillovers (Chain-Affiliated Facilities)}
\label{tab:spillover_chain_county}
\begin{threeparttable}
\footnotesize
\begin{tabular}{lccc}
\toprule
                    & (1) Overall rating & (2) Staffing rating & (3) Total deficiencies \\
\midrule
Chain peer mean     & 0.720   & 0.723   & 0.560   \\
                    & (0.020) & (0.019) & (0.029) \\
County peer mean    & 0.220   & 0.126   & 0.455   \\
                    & (0.024) & (0.017) & (0.032) \\
\midrule
Observations        & 8{,}791  & 8{,}757  & 8{,}852  \\
$R^{2}$             & 0.231    & 0.460    & 0.367    \\
\bottomrule
\end{tabular}
\begin{tablenotes}
\footnotesize
\item Notes: Each column reports a separate regression of the indicated outcome
on chain-level and county-level peer means, controlling for number of beds,
ownership, and state fixed effects. Sample restricted to chain-affiliated
facilities. Standard errors clustered at the chain level are in parentheses.
\end{tablenotes}
\end{threeparttable}
\end{table}

\begin{table}[!htbp]\centering
\caption{County-Level Threshold Estimates for Federal Oversight (SNF-Only)}
\label{tab:county_thresholds}
\begin{threeparttable}
\footnotesize
\begin{tabular}{lccc}
\toprule
 & (1) No Break 
 & (2) Break Only
 & (3) Break + Controls + State FE \\
\midrule
\textit{Break terms} \\
\quad $\log(n)$ 
    & $0.241^{***}$ & $0.120^{***}$ & $0.143^{***}$ \\
    & $(0.008)$     & $(0.007)$     & $(0.009)$     \\[0.25em]
\quad hinge$(\log(n)-2.20)$ 
    &                & $1.342^{***}$ & $1.381^{***}$ \\
    &                & $(0.128)$     & $(0.128)$     \\
\addlinespace
\textit{Controls} \\
share\_for\_profit
    &                &                & $-0.528^{**}$ \\
    &                &                & $(0.250)$     \\
share\_non\_profit
    &                &                & $-0.533^{**}$ \\
    &                &                & $(0.250)$     \\
share\_gov
    &                &                & $-0.553^{**}$ \\
    &                &                & $(0.251)$     \\
share\_in\_chain
    &                &                & $-0.020$      \\
    &                &                & $(0.016)$     \\
avg\_overall\_rating
    &                &                & $-0.004$      \\
    &                &                & $(0.005)$     \\
avg\_beds
    &                &                & $0.0002$      \\
    &                &                & $(0.0001)$    \\
avg\_def\_total
    &                &                & $0.011^{***}$\\
    &                &                & $(0.001)$    \\
State FE 
    & No             & No             & Yes            \\
\midrule
Observations  & 3{,}980 & 3{,}853 & 3{,}853 \\
RMSE          & 0.463   & 0.357   & 0.357   \\
\bottomrule
\end{tabular}
\begin{tablenotes}
\footnotesize
\item Notes: Dependent variable: number of SFF facilities in the county. Breakpoint fixed at $\hat c = 2.20$ (corresponding to $n^\ast \approx 7$ SNFs). Heteroskedasticity-robust standard errors in parentheses. ${}^{***} p<0.01$, ${}^{**} p<0.05$, ${}^{*} p<0.1$.
\end{tablenotes}
\end{threeparttable}
\end{table}

\begin{table}[!htbp]\centering
\caption{Chain-Level Threshold Estimates for Organizational Failures}
\label{tab:chain_thresholds}
\begin{threeparttable}
\footnotesize
\begin{tabular}{lccc}
\toprule
 & (1) No Break 
 & (2) Break Only
 & (3) Break + Controls \\
\midrule
\textit{Break terms} \\
\quad $\log(n)$ 
    & $0.142^{***}$ & $0.046^{**}$ & $0.046^{**}$ \\
    & $(0.017)$     & $(0.023)$    & $(0.021)$    \\[0.25em]
\quad hinge$(\log(n)-\log(34))$
    &                & $0.503^{***}$ & $0.503^{***}$ \\
    &                & $(0.209)$     & $(0.076)$     \\
\addlinespace
\textit{Controls} \\
share\_for\_profit
    &                &                & $0.225^{***}$ \\
    &                &                & $(0.073)$     \\
share\_non\_profit
    &                &                & $0.242^{***}$ \\
    &                &                & $(0.088)$     \\
share\_gov
    &                &                & $0.271^{**}$  \\
    &                &                & $(0.133)$     \\
avg\_overall\_rating
    &                &                & $-0.090^{***}$ \\
    &                &                & $(0.018)$      \\
State FE 
    & No             & No             & No             \\
\midrule
Observations  & 601 & 601 & 601 \\
RMSE          & 0.361 & 0.344 & 0.344 \\
\bottomrule
\end{tabular}
\begin{tablenotes}
\footnotesize
\item Notes: Dependent variable: number of SFF facilities in the chain. Breakpoint fixed at $\hat c = \log(34)$. HC1 robust standard errors in parentheses. ${}^{***} p<0.01$, ${}^{**} p<0.05$, ${}^{*} p<0.1$.
\end{tablenotes}
\end{threeparttable}
\end{table}

\begin{figure}[htbp]
\centering
\includegraphics[width=0.9\textwidth]{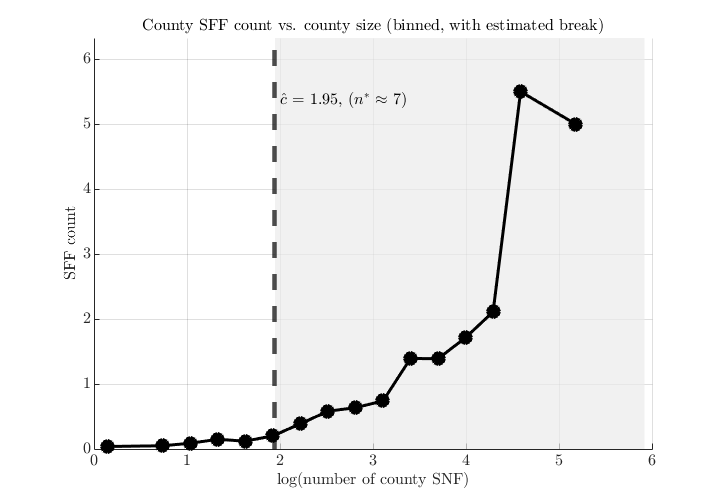}
\caption{Fitted Break Model for the SNF County Network. 
The plot shows binned averages of SFF incidence against $\log(n^{SNF}_c)$
and the fitted kinked regression line with estimated breakpoint 
$\hat c_{SNF} \approx 1.95$ ($n^{SNF*} \approx 7$).}
\label{fig:threshold_snf}
\end{figure}

\begin{figure}[htbp]
\centering
\includegraphics[width=0.9\textwidth]{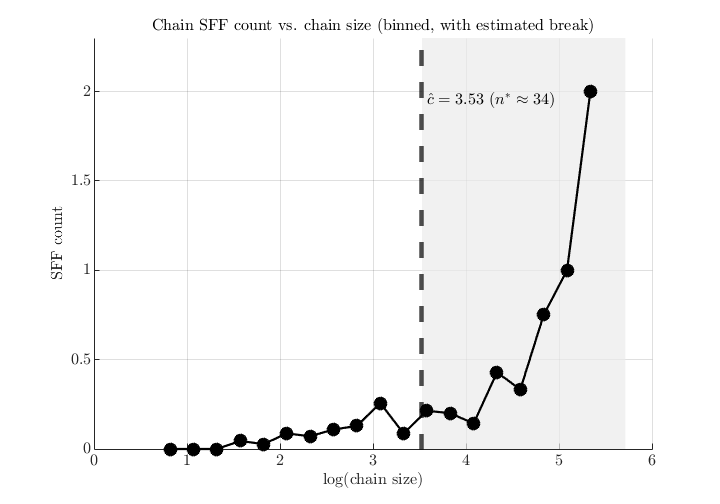}
\caption{Fitted Break Model for the Chain Network.
Binned scatter points show mean SFF incidence by $\log(n^{CHAIN}_j)$.
The fitted kinked regression uses the estimated breakpoint 
$\hat c_{CHAIN} = 3.526$ ($n^{CHAIN*} \approx 34$).}
\label{fig:threshold_chain}
\end{figure}

\begin{table}[htbp]\centering
\caption{County Peer Spillovers Above Size Cutoff}
\label{tab:spillover_county_above}
\begin{threeparttable}
\footnotesize
\begin{tabular}{lccc}
\toprule
                    & (1) Overall rating & (2) Staffing rating & (3) Total deficiencies \\
\midrule
County peer mean    & 0.757   & 0.484   & 0.812   \\
                    & (0.054) & (0.071) & (0.036) \\
\midrule
Observations        & 4{,}635  & 4{,}606  & 4{,}668  \\
$R^{2}$             & 0.211    & 0.330    & 0.340    \\
\bottomrule
\end{tabular}
\begin{tablenotes}
\footnotesize
\item Notes: Each column reports a separate regression of the indicated outcome
on the corresponding county peer mean, controlling for number of beds, ownership,
and state fixed effects. Sample restricted to counties above the estimated
network-size cutoff. Standard errors clustered at the county level are in parentheses.
\end{tablenotes}
\end{threeparttable}
\end{table}

\begin{table}[htbp]\centering
\caption{Chain and County Peer Spillovers Above Size Cutoffs (Chain-Affiliated Facilities)}
\label{tab:spillover_chain_county_above}
\begin{threeparttable}
\footnotesize
\begin{tabular}{lccc}
\toprule
                    & (1) Overall rating & (2) Staffing rating & (3) Total deficiencies \\
\midrule
Chain peer mean     & 0.836   & 0.707   & 0.496   \\
                    & (0.048) & (0.056) & (0.072) \\
County peer mean    & 0.256   & 0.112   & 0.509   \\
                    & (0.043) & (0.027) & (0.059) \\
\midrule
Observations        & 3{,}828  & 3{,}819  & 3{,}852  \\
$R^{2}$             & 0.188    & 0.347    & 0.365    \\
\bottomrule
\end{tabular}
\begin{tablenotes}
\footnotesize
\item Notes: Each column reports a separate regression of the indicated outcome
on chain-level and county-level peer means, controlling for number of beds,
ownership, and state fixed effects. Sample restricted to chain-affiliated facilities
in networks above the estimated size cutoffs. Standard errors clustered at the
chain level are in parentheses.
\end{tablenotes}
\end{threeparttable}
\end{table}

\begin{table}[htbp]\centering
\caption{County Peer Spillovers Below Size Cutoff}
\label{tab:below_cty_spill}
\begin{threeparttable}
\footnotesize
\begin{tabular}{lccc}
\toprule
                    & (1) Overall rating & (2) Staffing rating & (3) Total deficiencies \\
\midrule
County peer mean    & 0.150   & 0.099   & 0.272   \\
                    & (0.025) & (0.025) & (0.027) \\
\midrule
Observations        & 7{,}538  & 7{,}499  & 7{,}582  \\
$R^{2}$             & 0.117    & 0.324    & 0.236    \\
\bottomrule
\end{tabular}
\begin{tablenotes}
\footnotesize
\item Notes: Each column reports a separate regression of the indicated outcome
on the corresponding county peer mean, controlling for number of beds, ownership,
and state fixed effects. Sample restricted to counties at or below the estimated
network-size cutoff. Standard errors clustered at the county level are in parentheses.
\end{tablenotes}
\end{threeparttable}
\end{table}

\begin{table}[htbp]\centering
\caption{Chain and County Peer Spillovers Below Size Cutoffs (Chain-Affiliated Facilities)}
\label{tab:below_chain_cty_spill}
\begin{threeparttable}
\footnotesize
\begin{tabular}{lccc}
\toprule
                    & (1) Overall rating & (2) Staffing rating & (3) Total deficiencies \\
\midrule
Chain peer mean     & 0.680   & 0.723   & 0.579   \\
                    & (0.023) & (0.021) & (0.026) \\
County peer mean    & 0.193   & 0.138   & 0.426   \\
                    & (0.028) & (0.024) & (0.037) \\
\midrule
Observations        & 4{,}777  & 4{,}754  & 4{,}814  \\
$R^{2}$             & 0.282    & 0.527    & 0.386    \\
\bottomrule
\end{tabular}
\begin{tablenotes}
\footnotesize
\item Notes: Each column reports a separate regression of the indicated outcome
on chain-level and county-level peer means, controlling for number of beds,
ownership, and state fixed effects. Sample restricted to chain-affiliated facilities
in networks at or below the estimated size cutoffs. Standard errors clustered at the
chain level are in parentheses.
\end{tablenotes}
\end{threeparttable}
\end{table}

\begin{table}[htbp]\centering
\caption{Comparison of Spillover Peer Effects Above vs.\ Below the Cutoff}
\label{tab:spillover_comparison}
\begin{tabular}{lcccc}
\hline\hline
                    & \multicolumn{2}{c}{County Peer Effect} 
                    & \multicolumn{2}{c}{Chain Peer Effect} \\
\cline{2-3} \cline{4-5}
Outcome             & Above Cutoff & Below Cutoff 
                    & Above Cutoff & Below Cutoff \\ 
\hline
\textbf{Overall  rating}& 
0.757*** & 0.150*** & 
0.836*** & 0.680*** \\
 & (0.054) & (0.025) & (0.048) & (0.023) \\[0.6em]

\textbf{Staffing rating}& 
0.484*** & 0.099*** & 
0.707*** & 0.723*** \\
 & (0.071) & (0.025) & (0.056) & (0.021) \\[0.6em]

\textbf{Total deficiencies}& 
0.812*** & 0.272*** & 
0.496*** & 0.579*** \\
 & (0.036) & (0.027) & (0.072) & (0.026) \\[0.3em]

\hline
State FE included      & Yes & Yes & Yes & Yes \\
Controls included      & Yes & Yes & Yes & Yes \\
Cluster level          & County & County & Chain & Chain \\
\hline\hline
\multicolumn{5}{l}{\footnotesize Notes: Standard errors in parentheses. *** $p<0.01$.} \\
\end{tabular}
\end{table}

\begin{table}[htbp]\centering
\caption{Variance Comparison Across Network-Size Thresholds}
\label{tab:variance_thresholds}
\begin{threeparttable}
\footnotesize
\begin{tabular}{lcccc}
\toprule
 & \multicolumn{2}{c}{Variance} & \multicolumn{2}{c}{Levene Test} \\
\cmidrule(lr){2-3}\cmidrule(lr){4-5}
Outcome & Small network & Large network & $F$-stat & $p$-value \\
\midrule
\multicolumn{5}{l}{\textit{Panel A: Counties (threshold = 9 SNFs)}} \\
Overall rating       & 2.021  & 2.023  & 0.001   & 0.977      \\
Staffing rating      & 1.644  & 1.533  & 14.269  & 0.0001   \\
Total deficiencies   & 91.66 & 178.44& 361.496 & 0.000      \\
\addlinespace
\multicolumn{5}{l}{\textit{Panel B: Chains (threshold = 34 facilities)}} \\
Overall rating       & 1.998  & 1.836  & 26.979  & 0.000 \\
Staffing rating      & 1.539  & 1.082  & 134.35 & 0.000     \\
Total deficiencies   & 157.33& 135.04& 7.27   & 0.007    \\
\bottomrule
\end{tabular}
\begin{tablenotes}
\footnotesize
\item Notes: ``Small network'' denotes counties with $\leq 7$ SNFs (Panel A)
and chains with $\leq 34$ facilities (Panel B). ``Large network'' denotes
counties with $> 7$ SNFs and chains with $> 34$ facilities. Variances are
computed at the facility level. $p$-values are from Levene tests of equality
of variances across the small and large groups.
\end{tablenotes}
\end{threeparttable}
\end{table}

\begin{table}[htbp]\centering
\caption{Deterioration in Deficiencies and Network Size}
\label{tab:deterioration}
\begin{threeparttable}
\footnotesize
\begin{tabular}{lccc}
\toprule
 & Coefficient & 95\% CI (lower) & 95\% CI (upper) \\
\midrule
Constant ($\alpha$)              & -2.366 & -6.367  & 1.633  \\
Large chain ($\theta^{\text{ch}}$)   & 0.204  & -0.149 & 0.557 \\
Large county ($\theta^{\text{cty}}$) & 0.598  & 0.238  & 0.959 \\
\midrule
Observations                      & \multicolumn{3}{c}{14{,}675} \\
$R^{2}$                           & \multicolumn{3}{c}{0.085}    \\
$F$-statistic (overall)          & \multicolumn{3}{c}{23.863 \quad ($p = 5.012\times 10^{-235}$)} \\
\bottomrule
\end{tabular}
\begin{tablenotes}
\footnotesize
\item Notes: Dependent variable is the change in total deficiencies between
two inspection cycles, $\Delta \text{def}_i$. ``Large chain'' is an indicator
for facilities in chains with $>34$ facilities; ``large county'' is an
indicator for facilities in counties with $>7$ SNFs. The regression includes
facility controls and state fixed effects. Confidence intervals are based on
robust standard errors. The design matrix is rank deficient due to collinearity
among controls; redundant regressors are dropped automatically by the
estimation routine.
\end{tablenotes}
\end{threeparttable}
\end{table}

%=================================================
% Additional details on data cleaning and construction can be added here.
\begin{figure}[t]
    \centering
   
    \includegraphics[width=0.8\textwidth]{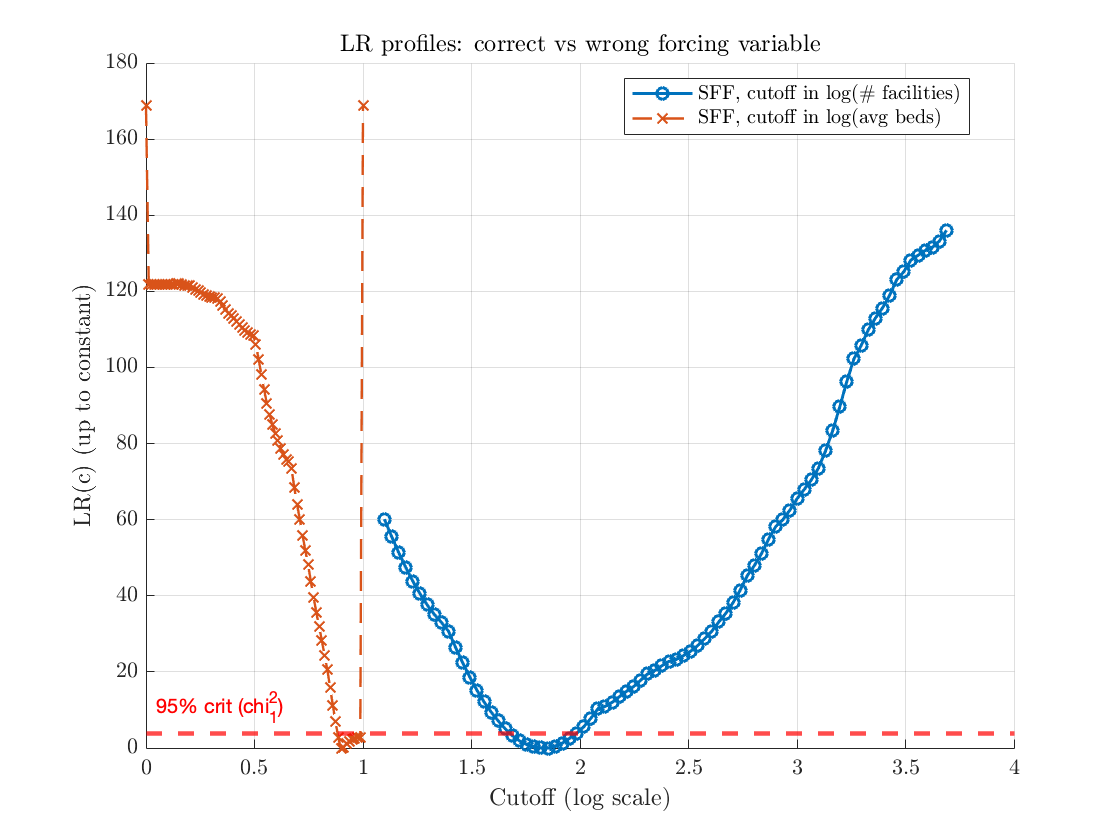}
    \caption{Likelihood-ratio profiles for county SFF thresholds using the correct 
    and a placebo forcing variable. The solid line plots the LR profile when the 
    forcing variable is $\log(\text{\# facilities})$, showing a clear interior 
    minimum around $\log n \approx 1.86$ ($n^\ast \approx 6$--$7$). The dashed line 
    plots the LR profile when the forcing variable is $\log(\text{average beds})$, 
    which is nearly monotone and attains its minimum at the boundary of the grid, 
    consistent with the absence of a meaningful threshold in this dimension.}
    \label{fig:LR_profile_county_placebo}
\end{figure}
\begin{table}[t]
\centering
\caption{Placebo threshold tests using alternative county outcomes}
\label{tab:placebo_wrong_outcome}
\begin{threeparttable}
\begin{tabular}{lccc}
\toprule
Outcome $y$ 
  & Forcing variable 
  & LR--mini cutoff $\hat c$ (log \# facilities) 
  & Implied $n^{\ast} = e^{\hat c}$ \\
\midrule
SFF count (baseline)       
  & $\log(\text{\# facilities})$ 
  & $1.843$ 
  & $6.32$ \\
Share non--profit (placebo 1) 
  & $\log(\text{\# facilities})$ 
  & $0.793$ 
  & $2.21$ \\
Share government (placebo 2) 
  & $\log(\text{\# facilities})$ 
  & $0.693$ 
  & $2.00$ \\
\bottomrule
\end{tabular}
\begin{tablenotes}[flushleft]\footnotesize
\item Notes: Each row reports the cutoff $\hat c$ that minimizes the residual sum of 
squares in a single--kink regression of the indicated outcome $y$ on 
$\log(\text{\# facilities})$, a kink term $(\log n - c)^{+}$, and the full set of 
county controls used in the baseline specification. The baseline row reproduces the 
main result for SFF counts. For the placebo outcomes, the LR profiles are relatively 
flat and the estimated cutoffs are very close to the lower bound of the grid 
($n^{\ast}\approx 2$ facilities), consistent with the absence of a meaningful 
network--size threshold in ownership composition.
\end{tablenotes}
\end{threeparttable}
\end{table}

\end{document}